\documentclass[lettersize,journal]{IEEEtran}
\usepackage{amsmath,amsfonts,mathtools}
\mathtoolsset{showonlyrefs}
\usepackage{algorithm}
\usepackage{algorithmic}
\usepackage{array}
\usepackage[caption=false,font=footnotesize]{subfig}
\usepackage{textcomp}
\usepackage{stfloats}
\usepackage{url}
\usepackage{moreverb}
\usepackage{verbatim}
\usepackage{graphicx}
\usepackage{cite}
\usepackage[hidelinks,colorlinks=false]{hyperref}
\usepackage{tikz}
\usepackage{centernot}
\usepackage{enumitem}
\usepackage{xcolor}

\usepackage{amsmath}
\usepackage{amsthm}
\usepackage{amssymb}
\usepackage{amsfonts}
\usepackage{xcolor}

\def\tr{^{\sf T}}

\def\nullcol{\mathrm{null}}
\def\spancol{\mathrm{span}}
\def\rankcol{\mathrm{rank}}

\newcommand{\inv}{^{-1}}

\newcommand{\mc}{\mathcal}

\newcommand\bbar[1]{\bar{\bar{#1}}}

\newtheorem{corollary}{Corollary}
\newtheorem{proposition}{Proposition}

\newtheorem{definition}{Definition}

\newtheorem{remark}{Remark}
\newtheorem{example}{Example}

\DeclareMathOperator*{\minimize}{minimize~}
\DeclareMathOperator*{\maximize}{maximize~}
\DeclareMathOperator*{\st}{subject\,to~}

\newcommand{\calQ}{\mathcal{Q}}

\newcommand{\R}{\mathbb{R}}

\newcommand{\dhdq}{\dfrac{\partial h}{\partial q}}
\newcommand{\dhdqsmall}{\frac{\partial h}{\partial q}}
\newcommand{\dhdqbar}{\dfrac{\partial \bar{h}}{\partial q}}
\newcommand{\dhdqbarsmall}{\frac{\partial \bar{h}}{\partial q}}

\begin{document}

\title{Extended Set-based Tasks\\for Multi-task Execution and Prioritization}

\author{%
    Gennaro Notomista, \IEEEmembership{Member, IEEE}, Mario Selvaggio, \IEEEmembership{Member, IEEE}, Francesca Pagano,\\ Mar\'{i}a Santos, \IEEEmembership{Member, IEEE}, Siddharth Mayya, Vincenzo Lippiello, \IEEEmembership{Senior Member, IEEE}, and Cristian Secchi, \IEEEmembership{Senior Member, IEEE}%
    \thanks{This work was supported in part by the NSERC Discovery Grant RGPIN-2023-03703 (G. Notomista).}%
    \thanks{G. Notomista is with the Department of Electrical and Computer Engineering, University of Waterloo, Waterloo, ON, Canada (email: {\tt gennaro.notomista@uwaterloo.ca}).}%
    \thanks{M. Selvaggio is with the Department of Electrical Engineering and Information Technology, University of Naples Federico II, Napoli, Italy (email: {\tt mario.selvaggio@unina.it}).}%
    \thanks{F. Pagano is with Centro Servizi Metrologici e Tecnologici Avanzati (CeSMA), Corso Nicolangelo Protopisani, 80146, Napoli, Italy (email: {\tt francesca.pagano@unina.it}).}%
    \thanks{M. Santos was with the Department of Mechanical and Aerospace Engineering, Princeton University, Princeton, NJ 08544, USA (email: {\tt maria.santos@princeton.edu}).}%
    \thanks{S. Mayya is with Aquatic Labs, Cambridge, MA 02139, USA (email: {\tt mayya.siddharth@gmail.com}). This work is not related to Aquatic Labs.}%
    \thanks{V. Lippiello is with the Department of Electrical Engineering and Information Technology, University of Naples Federico II, Napoli, Italy (email: {\tt vincenzo.lippiello@unina.it}).}%
    \thanks{C. Secchi is with the Department of Science and Methods of Engineering, University of Modena and Reggio Emilia, Modena, Italy (email: {\tt cristian.secchi@unimore.it}).}%
}

\maketitle

\begin{abstract}
The ability of executing multiple tasks simultaneously is an important feature of redundant robotic systems. As a matter of fact, complex behaviors can often be obtained as a result of the execution of several tasks. Moreover, in safety-critical applications, tasks designed to ensure the safety of the robot and its surroundings have to be executed along with other nominal tasks. In such cases, it is also important to prioritize the former over the latter. In this paper, we formalize the definition of extended set-based tasks, i.e., tasks which can be executed by rendering subsets of the task space asymptotically stable or forward invariant using control barrier functions. We propose a formal mathematical representation of such tasks that allows for the execution of more complex and time-varying prioritized stacks of tasks using kinematic and dynamic robot models alike. We present an optimization-based framework which is computationally efficient, accounts for input bounds, and allows for the stable execution of time-varying prioritized stacks of extended set-based tasks. The proposed framework is validated using extensive simulations, quantitative comparisons to the state-of-the-art hierarchical quadratic programming, and experiments with robotic manipulators.
\end{abstract}

\section{Introduction}
\label{sec:intro}

Safe and effective operation of robotic systems requires the simultaneous execution of multiple tasks. This concurrence, which stems from both application and safety needs, can be achieved thanks to the \emph{redundancy} of a robotic system, which affords the execution of a task in multiple ways. In the case of a robot manipulator arm, for instance, kinematic redundancy consists in having more degrees of freedom (DOFs) than the ones strictly required to execute the task. This implies the existence of different joint velocities (or configurations) that result in the same end-effector velocity (or configuration). For a multi-robot system, the robotic units of which it is comprised and their interchangeability translate into an inherent redundancy of the system. For an extended discussion on redundancy see, for example,~\cite{milutinovic2013redundancy,siciliano2010robotics}.

During the operation of the robotic system, some tasks may take precedence over others. Typically, preserving the safety of the system and its surroundings is more critical than satisfying application-specific objectives. This hierarchy is established in the so-called \textit{task stack}, used by the robot controller to ensure the prioritized execution of tasks. This way, while high-priority tasks are being executed, the redundancy of the system is exploited to achieve, if possible, lower-priority ones. Moreover, the task hierarchy need not be static. For instance, this can be needed to accommodate the evolution of the objectives dictated by the environment/application, or due to exogenous factors (such as the supervisory control of a human operator). In any case, the robot controller must be able to react to these changes.

The execution of prioritized stacks of tasks by exploiting the redundancy of the system has been proposed in~\cite{siciliano1991} and extensively studied since then (see, e.g.,~\cite{baerlocher2004inverse, escande2014hierarchical}, just to name a few). In these works, a task is associated to a Jacobian matrix, which relates the velocities of the robot joints to the velocities of the operational point of the robot kinematic structure required to execute the task. While the resulting formulation is rigorous, starting from the definition of a Jacobian matrix in order to specify a task is not always intuitive. In~\cite{notomista2020set}, extended set-based (ESB) tasks---a more expressive description of tasks which leverages set stability and safety as building blocks to define robotic tasks---are introduced and an optimization-based approach for their prioritized execution is proposed.

In this paper, we perform an extensive analysis of the (possibly time-varying) prioritized execution of ESB tasks. Moreover, we propose a novel optimization-based approach for the task execution and apply it to both kinematic and dynamic robot models, allowing torque control of robotic manipulators, and accounting for input constraints, such as torque saturation. 

\subsection{Related Work}

The name \textit{task function}, which we adopt in this paper, has appeared in \cite{samson1991robot}, where an approach to represent robotic tasks based on mappings from the configuration space to the task space is introduced. This seminal work laid the foundations for several methods developed in the past years for robot control which are based on such task description (see, e.g., \cite{fiore2023convergence,kermorgant2013dealing,MarchandIROS96}, to cite a few). The approach proposed in \cite{samson1991robot} encodes the task execution via the stability of a linear differential equation describing the behavior of the task function. The notion of ESB task functions---defined in \cite{notomista2020set}---we adopt in this paper differs from the line of work based on~\cite{samson1991robot} for the following main reasons: (i) ESB tasks provide a way to encode also safety-like tasks using control barrier functions; (ii) The execution of ESB tasks is achieved via differential inequality, rather than equalities, which allows for more flexible control selection. 

In applications, it is common to assign different priorities to tasks, thus requiring control algorithms capable of accounting for them~\cite{nakamura1987task}. Methods for redundancy resolution with prioritized tasks have been thoroughly developed in the last decades. At the differential kinematic level, as an infinite number of joint velocities exist that realize a given task velocity, a criterion must be utilized to select one of them. The use of the (right) pseudoinverse of the Jacobian matrix guarantees the exact reconstruction of the task velocity with the minimum-norm joint velocities. Most methods calculate joint velocities as the general solution of an undetermined linear system, i.e., using the generalized inverse of the task Jacobian plus a velocity vector lying in its null-space. Building upon this concept, a general framework for managing multiple tasks in highly redundant robotic systems, exploiting their kinematic redundancies is presented in~\cite{siciliano1991}.

As mentioned in the previous section, the prioritization approach based on tasks described using Jacobian matrices is rigorous and it has been successfully employed for different robotic systems, including robotic arms, legged robots, and multi-robot systems. Relatively more recently, however, different paradigms have been proposed in order to achieve the concurrent execution of multiple tasks, including the null-space-based behavioral control~\cite{antonelli2009tcns} and the set-based tasks, i.e., tasks with a range of valid values~\cite{moe2016set}, where the authors give the conditions under which the concurrent execution of multiple tasks is guaranteed. Moreover, numerical optimization-based approaches started to be used for prioritized task executions~\cite{escande2014hierarchical, kanoun2011kinematic}. 

Approaches to enforce a hierarchy among the executed tasks can then be categorized into \emph{strict} and \emph{soft}. Strict hierarchies guarantee that the execution of a task at lower priority does not influence the execution of a task at higher priority. More recently, CBF-Based Hierarchical Quadratic Programs~\cite{Xie2025} have been proposed to address prioritization of both safety and performance objectives. This, however, do not explicitly address the problem of switching stacks and, therefore, may lead to discontinuities of the robot controller during the switch between stacks of tasks. Several works have proposed approaches to handle these discontinuities (including practical ones) using techniques suitable to the considered task and stack descriptions~\cite{antonelli2009tcns, keith2011analysis, simetti2016novel}. The Hierarchical Quadratic Programming (HQP) approach presented in~\cite{kim2019continuous} allows for continuous, time-based stack switching, and it is amenable to torque control. 

Beyond prescribed temporal switching, state-dependent task activation is addressed in~\cite{lee2012intermediate} with the intermediate desired value approach, which adjusts task references for continuous insertion, removal, and priority rearrangement. For underwater manipulation, state-dependent activation and regularized inverses are proposed in~\cite{simetti2014floating}; iCAT~\cite{simetti2016novel} extends this mechanism to individual rows of multidimensional tasks, mitigating practical discontinuities and task interference. Building on iCAT, damped projectors for continuous task and subtask activation and deactivation in null-space-based fleet control are used in~\cite{mannucci2021null}. 

The soft task hierarchy enforcement, instead, can be leveraged to guarantee continuous transitions between stacks of tasks, at the expense of letting tasks at lower priorities influence tasks at higher priorities~\cite{SilverioTRO2019}. In this paper, we opt for a soft task hierarchy approach, in order to reduce the computational complexity of the execution and prioritization of tasks, as well as of the stack transitions. We also establish conditions under which lower-priority tasks do not affect the execution of higher-priority ones. Our switching analysis considers time-parameterized transitions; extending the stability results to state-dependent transitions is left for future work.

\subsection{Contributions}
The main contributions of this paper compared to the state of the art are summarized as follows:
\begin{enumerate}[label=(\roman*)]
	\item We introduce inter-task relationships between ESB tasks and compare them with the ones defined for Jacobian-based tasks.
	\item We provide the stability analysis of the prioritized execution of ESB tasks, even in presence of switching stacks of tasks.
	\item We show how the proposed control framework is applicable to torque-controlled robots.
	\item We compare our proposed approach with a state-of-the-art HQP-based stack switching method and show its computational advantage over the latter.
\end{enumerate}
In particular, all the above contributions constitute advancements compared to our previous work \cite{notomista2020set}.

\section{Extended Set-based Task Execution}
\label{sec:tasks}

Consider a $N$-DOF manipulator and let $q\in\mc Q$ denote the joint configuration, an element of the $N$-dimensional configuration space $\mc Q$, and $\dot q \in\R^N$, the joint velocity vector. Let $m$ denote the dimension of the task space $\mc T$. We start by considering robots that can be modeled as kinematic systems, i.e., that are endowed with a low-level high-gain controller that takes care of reproducing the desired velocity inputs. Let $\sigma \in \mc T $ denote the task variable to be controlled, defined as:
\begin{equation}
	\label{eq:sigma}
	\sigma = k (q),
\end{equation}
where the smooth map $k:\mc Q\rightarrow \mc T$ represents the task forward kinematics. 
The novelty of this paper consists in the multi-task execution and prioritization framework presented in the following section. From \eqref{eq:sigma}, it follows that
\begin{equation}
	\label{eq:robotkinmodel}
	\dot \sigma = J(q) \dot q,
\end{equation}
where $J\colon \mc Q\to\R^{m\times N}$ is the task Jacobian defined as $J(q)=\frac{\partial k}{\partial q}(q)$. The tasks consisting in executing the joint velocities $\dot q$ satisfying the relation in \eqref{eq:robotkinmodel} to make $\sigma(t)$ equal to a desired task variable trajectory $\sigma_d(t)$ will be referred to as \emph{Jacobian-based tasks}. A Jacobian-based task is accomplished when the task variable $\sigma(t)$ tracks the desired $\sigma_d(t)$.

As discussed in Sec.~\ref{sec:intro}, there is often a need to concurrently execute multiple Jacobian-based tasks. Towards this end, let $\sigma_i=k_i(q) \in\R^{m_i}$---the task space is an $m_i$-dimensional Euclidean space---be the task variables associated to $M$ tasks, where $i=\{1,\dots ,M\}$. In the following, we will use $m$ to denote the sum of dimensions of task variables, $m = \sum_{i=1}^M m_i$. The evolution of each task variable can then be represented as $\dot \sigma_i = J_i(q) \dot q$, where $J_i(q)=\frac{\partial k_i}{\partial q}(q)$ is the $i$-th task Jacobian. For more details on Jacobian-based tasks, we address the interested reader to, e.g.,~\cite{antonelli2009tro} and references therein. Next, we introduce the concept of ESB tasks.

When considering Jacobian-based tasks, the main approach for executing a task is to drive the task variable towards a desired value (see~\cite{antonelli2009tro}). Nevertheless, this approach defines tasks based on the velocities of the task variable, $\dot\sigma$. It is often more convenient to define tasks that can be accomplished by making the task variable approach a desired set of the task space (stability) or constraining the task variable to remain within a given set (safety). While the Jacobian-based tasks do not lend themselves to this type of task definition, the ESB tasks used in this paper do.

Set-based tasks have been introduced in~\cite{escande2014hierarchical} for representing tasks that can be executed by letting the task variables remain in a desired area of satisfaction. Some examples of set-based tasks are joint limits or singularity avoidance, where the area of satisfaction is not a specific configuration but rather the set of joint values that are not including joint-limit values or singularity configurations, respectively.

There is an analogy between the set-based interpretation of tasks and the concept of \emph{forward invariance} in the dynamical systems literature~\cite{khalil2015nonlinear}. In fact, executing a set-based task is equivalent to enforcing the state of the robot to remain in a desired set, i.e., making the set forward invariant. Exploiting this analogy, it is possible to extend the definition of set-based tasks in order to include the possibility of executing a task evolving towards a set, which corresponds to the concept of set-stability.

\begin{definition}[ESB Task~\cite{notomista2020set}]\label{def:esbt}
	An \textit{ESB task} is a task characterized by a set $\mc C \subset \mc T$, where $\mc T$ is the task space, which can be expressed as the zero superlevel set of a continuously differentiable function $h \colon \mc T \times \R_{\ge0} \to \R$ as follows:
	\begin{equation}
		\label{eq:defSafeset}
		\mc C = \{ \sigma \in \mc T \colon h(\sigma,t)\ge0 \}.
	\end{equation}
	The goal is rendering the time-varying set forward invariant and asymptotically stable.
\end{definition}

ESB tasks generalize Jacobian-based ones as the set $\mc C$ in~\eqref{eq:defSafeset} can be defined starting from $k(q)$ in \eqref{eq:sigma}, as partially illustrated in \cite{notomista2020set}. In the latter, however, the authors did not analyze the inter-ESB-task relationships. In Sec.~\ref{subsec:intertaskrel}, we will analyze these relationships and show how they extend the ones between Jacobian-based tasks. It is worthwhile to point out that even starting from equivalent definitions of Jacobian-based and ESB tasks, their execution may differ significantly from a quantitative point of view, due to the different control synthesis frameworks typically adopted with the respective task formulations.

Given a task $\sigma = k(q)$, the following control affine system can be employed to define the evolution of the robot as well as the task:
\begin{equation}
    \label{eq:ca-dyn}
	\begin{cases}
		\dot x = f(x) + g(x)u\\
		\sigma = k(x),
	\end{cases}
\end{equation}
where, as before, the smooth map $k:\mc Q\rightarrow \mc T$ represents the task forward kinematics. 
This control affine form encompasses the kinematic robot model \eqref{eq:robotkinmodel} by setting $x = q$, $f = 0$, $g = I$ and $u = \dot q$. The standard dynamic model of a robot manipulator is also affine in the control input (joint torques), and can be therefore represented by \eqref{eq:ca-dyn} (see Sec.~\ref{sec:dynamics}).

Starting from the methodology for task execution proposed in \cite{notomista2018constraint}, we now develop a strategy for the execution of ESB tasks. According to Definition~\ref{def:esbt}, in order to accomplish an ESB task, it is necessary to control the robot for making the satisfaction set $\mc C$ attractive and forward invariant. As shown in~\cite{notomista2018constraint,notomista2019optimal}, this behavior can be formulated as a constrained optimization problem where the control input is chosen to enforce the task satisfaction. This desired behavior can be formalized using Control Barrier Functions (CBFs) that render the desired set forward invariant and stable (see~\cite{ames2019control,Ferraguti2022}).

To execute multiple tasks concurrently, each task is specified as a constraint in an optimization problem. Formally, let $T_1, \dots, T_M $ be the tasks to be executed, with $\sigma_i$ denoting the task variable corresponding to task $T_i$. The execution of task $T_i$ is encoded by the CBF $h_i,~i\in\{1,\dots , M\}$. The execution of the set of $M$ tasks can be realized through the solution to the following optimization problem:
\begin{equation}
	\label{eq:mainQP}
	\begin{aligned}
		\minimize_{u,\delta} &\|u\|^2 + l\|\delta\|^2 \\
		\st & \frac{\partial h_i}{\partial t} + \frac{\partial h_i}{\partial \sigma_i} \frac{\partial \sigma_i}{\partial x} f(x) + \frac{\partial h_i}{\partial \sigma_i} \frac{\partial \sigma_i}{\partial x} g(x) u \\
		&+ \gamma_i(h_i(\sigma_i,t)) \ge - \delta_{i}\quad\forall i \in \{1,\ldots,M\},
	\end{aligned}
\end{equation}
where $l>0$, $\gamma_i~\forall i \in \{1,\ldots,M\}$ are differentiable extended class $\mc K_\infty$ functions. Notice how each task is associated with a scalar slack variable $\delta_i$, with $\delta = [\delta_1, \ldots, \delta_M]\tr$. This allows us, first of all, to ensure the feasibility of the optimization program~\eqref{eq:mainQP}. In fact, consider the case when the $M$ tasks cannot be executed concurrently, i.e., the half-spaces defined by the affine inequality constraints in~\eqref{eq:mainQP} without $\delta_i$ do not intersect. If slack variables are removed from the constraints, then the optimization program may become infeasible. By employing slack variables to relax task constraints and adding the term $l\|\delta\|^2$ in the cost of~\eqref{eq:mainQP}, the control input solution of the optimization program results in the execution of multiple tasks to the best of the capabilities of a robotic system.

In the following section, we analyze the relationship between tasks that characterizes the feasibility of their concurrent execution without the need of slack variables. Moreover, as will be shown in Sec.~\ref{sec:priorities}, slack variables will be leveraged to enforce relative priorities between tasks.

\section{Prioritized Multi-task Execution}
\label{sec:priorities}

As discussed in Sec.~\ref{sec:intro}, to effectively handle the simultaneous execution of tasks which cannot be executed concurrently, we introduced slack variables in the optimization program~\eqref{eq:mainQP} used to synthesize the robot controller. As shown in this section, the introduction of slack variables can be used to enforce prioritization among tasks. This section extends the multi-task execution framework presented in the previous section to introduce prioritized execution and proves the stability properties of such a framework.

\subsection{Analyzing Inter-task Relationships}
\label{subsec:intertaskrel}

In order to evaluate if a set of tasks can be executed concurrently, it is necessary to characterize the interdependence properties of the tasks, intended as in Definition~\ref{def:esbt}, i.e., ESB tasks. In this section, the concepts of orthogonality, independence, and dependence---commonly used in the literature for describing the relationships between Jacobian-based tasks---are extended to ESB tasks. In the following, we first recall these concepts for Jacobian-based tasks.

Consider two Jacobian-based tasks encoded by $\sigma_i(q) \in \R^{m_i}$, $\sigma_j(q) \in \R^{m_j}$, whose velocities are expressed by the mappings $\dot{\sigma}_i(q) = J_i(q)\dot{q}$, $\dot{\sigma}_j(q) = J_j(q)\dot{q}$, with $J_i\neq0$ and $J_j\neq0$. The following definitions are used to recall the three distinct types of inter-task relationships discussed in the literature.

\begin{definition}[Jacobian-based task orthogonality~\cite{antonelli2009tro}]\label{def:orthogonal_J_tasks}
	Two Jacobian-based tasks $\sigma_i(q)$ and $\sigma_j(q)$ are orthogonal (or annihilating) if $J_i(q)J_j(q)\tr = 0_{m_i \times m_j},~\forall q \in \calQ$,
	where $0_{m_i \times m_j}$ is the $m_i \times m_j$ null matrix.
\end{definition}

\begin{definition}[Jacobian-based task independence~\cite{antonelli2009tro}]\label{def:independent_J_tasks} 
	Two Jacobian-based tasks $\sigma_i(q)$ and $\sigma_j(q)$ are independent if $\rankcol(J_i\tr(q)) + \rankcol(J_j\tr(q)) = \rankcol\left(\begin{bmatrix}J_i\tr(q) & J_j\tr(q)\end{bmatrix}\right),~\forall q \in \calQ$.
\end{definition}

\begin{definition}[Jacobian-based task dependence~\cite{antonelli2009tro}]\label{def:dependent_J_tasks}
	Two Jacobian-based tasks $\sigma_i(q)$ and $\sigma_j(q)$ are dependent if $\exists~q \in \calQ ~\colon~ \rankcol(J_i(q)\tr) + \rankcol(J_j(q)\tr) > \rankcol\left(\begin{bmatrix}J_i(q)\tr J_j(q)\tr\end{bmatrix}\right).$
\end{definition}

\begin{remark}
	Orthogonal Jacobian-based tasks are independent.
\end{remark}

\begin{remark}
	\label{rmk:pseudoinversetranspose}
	It is worth noting that the three conditions of orthogonality, dependence, and independence may be given by resorting to the pseudoinverse of the corresponding Jacobians instead of the transpose. In fact, they share the same span.
\end{remark}

We now derive the conditions required to characterize the inter-task relationships of orthogonality, dependence, and independence for ESB tasks. The definitions in the following are introduced to strengthen and complete the parallelism between the Jacobian-based and ESB-based task relationships, which has only been recognized but, to the best of our knowledge, never analyzed in the literature. Consider two ESB tasks encoded by the CBFs $h_i$ and $h_j$, with $\frac{\partial h_i}{\partial q}\neq0$ and $\frac{\partial h_j}{\partial q}\neq0$.

\begin{definition}[ESB task orthogonality]\label{def:orthogonal_CBF_tasks}
	Two ESB tasks encoded by the CBFs $h_i$ and $h_j$ are orthogonal if 
	\begin{equation}\label{eq:orthogonal_CBF_tasks}
		\frac{\partial h_i}{\partial q}\frac{\partial h_j}{\partial q}\tr = 0, \quad \forall q \in \calQ.
	\end{equation}
\end{definition}

Definition~\ref{def:orthogonal_CBF_tasks} is well posed since it generalizes the concept of orthogonality in Definition~\ref{def:orthogonal_J_tasks}. This is shown by the following result:

\begin{proposition}[Generalization of task orthogonality]
	\label{prop:esbtaskorthogonality}
	If two Jacobian-based tasks $\dot\sigma_i = J_i(q)\dot q$ and $\dot\sigma_j = J_j(q)\dot q$ are orthogonal according to Definition~\ref{def:orthogonal_J_tasks}, the corresponding ESB tasks $h_i$ and $h_j$ constructed as in~\cite{notomista2020set} are orthogonal according to Definition~\ref{def:orthogonal_CBF_tasks}. The converse is not necessarily true.
\end{proposition}
\begin{proof}
	To prove that $J_i(q) J_j(q)\tr = 0 \implies
	\frac{\partial h_i}{\partial q}\frac{\partial h_j}{\partial q}\tr = 0$,
	we expand~\eqref{eq:orthogonal_CBF_tasks} as follows: $\frac{\partial h_i}{\partial q}\frac{\partial h_j}{\partial q}\tr = \frac{\partial h_i}{\partial \sigma_i} J_i(q) J_j(q)\tr \frac{\partial h_j}{\partial \sigma_j}\tr = 0$.
	One can see how, whenever $J_i(q)$ or $J\tr_j(q)$ have a non-trivial nullspace, the converse need not to hold true.
\end{proof}
Proposition~\ref{prop:esbtaskorthogonality} shows how ESB tasks can generalize Jacobian-based tasks. This will be confirmed by the following results on dependence and independence.

\begin{definition}[ESB task independence]\label{def:independent_CBF_tasks}
	Two ESB tasks encoded by the CBFs $h_i$ and $h_j$ are independent if 
	\begin{equation}\label{eq:independent_CBF_tasks}
		\rankcol\left(\begin{bmatrix}\dfrac{\partial h_i}{\partial q}\tr & \dfrac{\partial h_j}{\partial q}\tr\end{bmatrix}\right) = 2, \quad \forall q \in \calQ.
	\end{equation}
\end{definition}

\begin{remark}
	\label{rmk:orthogonal}
	Orthogonal ESB tasks are independent.
\end{remark}

Definition~\ref{def:independent_CBF_tasks} is well posed since it generalizes Definition~\ref{def:independent_J_tasks} as shown by next proposition.

\begin{proposition}[Generalization of task independence]\label{prop:gener_task_indep}
	If two Jacobian-based tasks $\dot\sigma_i = J_i(q)\dot q$ and $\dot\sigma_j = J_j(q)\dot q$ are independent according to Definition~\ref{def:independent_J_tasks}, the corresponding ESB tasks $h_i(\sigma_i(q))$ and $h_j(\sigma_j(q))$ are independent according to Definition~\ref{def:independent_CBF_tasks}). The converse is not necessarily true.
\end{proposition}
\begin{proof}
	To prove that $\rankcol(J_i(q)\tr) + \rankcol(J_j(q)\tr) = \rankcol\left(\begin{bmatrix}J_i(q)\tr J_j(q)\tr\end{bmatrix}\right) \implies \rankcol\left(\begin{bmatrix}\frac{\partial h_i}{\partial q}\tr & \frac{\partial h_j}{\partial q}\tr\end{bmatrix}\right) = 2$,
	recall that subspace independence means $\spancol(J_i(q)\tr) \cap \spancol(J_j(q)\tr) = \{0\}$.
	Since $\frac{\partial h_i}{\partial q}\tr = J_i\tr\frac{\partial h_i}{\partial \sigma_i}\tr \in \spancol(J_i(q)\tr)$
	and $\frac{\partial h_j}{\partial q}\tr = J_j\tr\frac{\partial h_j}{\partial \sigma_j}\tr \in \spancol(J_j(q)\tr)$,
	we can claim independence of the vectors $\frac{\partial h_i}{\partial q}\tr$ and $\frac{\partial h_j}{\partial q}\tr$. The converse is not necessarily true whenever $J\tr_i(q)$ or $J\tr_j(q)$ have a non-trivial nullspace.
\end{proof}

Finally, the task dependence concept can be also generalized to ESB tasks.

\begin{definition}[ESB task dependence]\label{def:dependent_CBF_tasks}
	Two ESB tasks encoded by the CBFs $h_i$ and $h_j$ are dependent if $\exists \, q \in \calQ$ stuch that $\rankcol\left(\begin{bmatrix}\dfrac{\partial h_i}{\partial q}\tr & \dfrac{\partial h_j}{\partial q}\tr\end{bmatrix}\right) = 1$.
\end{definition}

\begin{remark}
	The task orthogonality condition~\eqref{eq:orthogonal_CBF_tasks} corresponds to the geometric condition $\angle\left(\frac{\partial h_i}{\partial q}\tr,\frac{\partial h_j}{\partial q}\tr\right) = \frac{\pi}{2} + k\pi$.
	The task independence condition~\eqref{eq:independent_CBF_tasks} corresponds to the geometric condition $\angle\left(\frac{\partial h_i}{\partial q}\tr,\frac{\partial h_j}{\partial q}\tr\right) \in \left(0,\pi\right)$,
	while dependence corresponds to $\angle\left(\frac{\partial h_i}{\partial q}\tr,\frac{\partial h_j}{\partial q}\tr\right) = k\pi$.
\end{remark}

The above presented results for ESB tasks and their connections with Jacobian-based tasks are summarized in Table~\ref{tab:task_interdep}.
\begin{table*}
	\caption{Inter-task Relationships: orthogonality, independence, and dependence for Jacobian-based and ESB tasks}
	\centering
	\begin{tabular}{lcc}
		\hline\hline
		&&\\[-0.25cm]
		Relationship & Jacobian-based tasks & ESB tasks \\
		\hline\hline
		&&\\[-0.2cm]
		Orthogonality & $J_i J_j\tr = 0_{m_i \times m_j}$ & $\dfrac{\partial h_i}{\partial q}\dfrac{\partial h_j}{\partial q}\tr = 0$ \\ 
		&&\\[-0.2cm]\hline&&\\[-0.2cm]
		Independence & $\rankcol(J_i\tr) + \rankcol(J_j\tr) = \rankcol\left(\begin{bmatrix}J_i\tr & J_j\tr\end{bmatrix}\right)$ & $\rankcol\left(\begin{bmatrix}\dfrac{\partial h_i}{\partial q}\tr & \dfrac{\partial h_j}{\partial q}\tr\end{bmatrix}\right) = 2$ \\
		&&\\[-0.2cm]\hline&&\\[-0.2cm]
		Dependence & $\exists \, q \in \calQ ~\colon~ \rankcol(J_i\tr) + \rankcol(J_j\tr) > \rankcol\left(\begin{bmatrix}J_i\tr J_j\tr\end{bmatrix}\right)$ & $\exists \, q \in \calQ ~\colon~ \rankcol\left(\begin{bmatrix}\dfrac{\partial h_i}{\partial q}\tr \dfrac{\partial h_j}{\partial q}\tr\end{bmatrix}\right) = 1$ \\
		&&\\[-0.2cm]\hline \hline
	\end{tabular}
	\label{tab:task_interdep}
\end{table*}

\subsection{Prioritized Execution of Extended Set-based Task Stacks}
\label{subsec:prioritization}
Within the multi-task execution framework introduced in~\eqref{eq:mainQP}, a natural way to introduce priorities is to enforce relative constraints among the slack variables corresponding to the different tasks that need to be performed, as proposed in~\cite{notomista2020set}. For example, if the robot were to perform task $T_i$ with the highest priority (often referred to using the partial order notation $T_i\prec T_j$ $\forall j \in \{1,\ldots,M\}, j\neq i$) the additional constraint 
\begin{equation}
\label{eq:delta_i}
\delta_i \leq \delta_j/\kappa, \quad \kappa>1, \, \forall j \in \{1,\ldots,M\}, \, j\neq i
\end{equation}
in~\eqref{eq:mainQP} would imply that task $T_i$ is relaxed to a lesser extent---thus performed with a higher priority---than the other tasks. The relative scale between the functions $h_i$ encoding the tasks should be considered while choosing the value of $\kappa$.

Formally, prioritizations among tasks can be encoded through an additional linear constraint $K\delta \leq 0$ to be enforced in the optimization problem~\eqref{eq:mainQP}. ${K\in\R^{(M-1) \times M}}$ is referred to as the \textit{prioritization matrix} and it encodes the pairwise inequality constraints among the slack variables, thus fully specifying the \textit{prioritization stack} among the tasks.

For $M$ extended-set based tasks with task variables $\sigma_i$ and whose execution is characterized by the CBFs $h_i$, $i\in\{1,\ldots,M\}$, the execution of the prioritized task stack can then be encoded via the following optimization problem which features the prioritization matrix $K$ defined as above:
\begin{equation}
	\label{eq:mainQPpriorities}
	\begin{aligned}
		\minimize_{u,\delta} &\|u\|^2 + l\|\delta\|^2 \\
		\st & \frac{\partial h_i}{\partial t} + \frac{\partial h_i}{\partial \sigma_i} \frac{\partial \sigma_i}{\partial x} f(x) + \frac{\partial h_i}{\partial \sigma_i} \frac{\partial \sigma_i}{\partial x} g(x) u \\
		&+ \gamma_i(h_i(\sigma_i,t)) \ge - \delta_{i},\quad\forall i \in \{1,\ldots,M\}\\
		&K\delta\le0.
	\end{aligned}
\end{equation}

Nevertheless, since these constraints represent tasks to be executed by a robot, we are concerned with the quality of execution of these tasks. A large value of $\delta_i$, as a matter of fact, might result in task $i$ not being executed. Therefore, we would like the values of $\delta_i$ to be as small as possible, possibly zero, at least when task $i$ has highest priority. Because of the prioritization constraint $K\delta\le0$, however, the component of $\delta$ corresponding to the task at highest priority is \textit{constrained} by the component of $\delta$ corresponding to the tasks at lower priorities. 
By increasing the constant $l$ in~\eqref{eq:mainQPpriorities} the computed solution better approximates the desired behavior, i.e., tasks at the highest priority being perfectly executed.

\begin{remark}[Feasibility of the prioritization stack]
\label{rmk:stackfeasibility}
The constraint $K\delta\le0$ imposes an upper bound on the ratios among the components of $\delta$. However, when a sequential priority stack $T_i \prec T_j \prec T_k$ is enforced, and tasks are conflicting, the values of the slack variables grow along the stack as a power of $\kappa$, as dictated by the inequality in~\eqref{eq:delta_i}. Although problem~\eqref{eq:mainQPpriorities} remains feasible, this can make it numerically ill-conditioned.
This issue can be overcome by appropriately choosing different gains $\kappa$ of the prioritization matrix $K$, or by enforcing the task hierarchy $T_i \prec T_j, T_k$, which, however, differs from the original one.
In the following section, we present an alternative mechanism to tackle this issue based on the relaxation of the prioritization constraint, which overcomes the limitation of the approach in~\cite{notomista2020set}.
\end{remark}

\subsection{Extended Set-based Task Execution with Torque Bounds}
\label{sec:dynamics}

In the previous sections, we analyzed the execution of multiple prioritized tasks using a kinematic model of robotic manipulators where there is a nonlinear single-integrator relation---given by \eqref{eq:robotkinmodel}---between the velocity in the task space and the velocity in the joint space. In this section, we extend the approach developed so far to robot dynamic models.

In task-prioritized control of robotic manipulators executing highly dynamic tasks---i.e., tasks requiring large accelerations, so that the effects of robot inertia are not negligible anymore---the following robot dynamic model is typically employed:
\begin{equation}
	\label{eq:robotdynmodel}
	D(q) \ddot q + C(q,\dot q) \dot q + F_v \dot q + g(q) = \tau,
\end{equation}
where $q$ are the joint angles, $D(q)$ is the inertia matrix, $C(q,\dot q)$ accounts for centrifugal and Coriolis effects, $F_v \dot q$ models viscous friction (static friction is neglected), $g(q)$ is the effect of gravity, and $\tau$ is the vector of joint input torques (see, e.g.,\cite{siciliano2010robotics}).

The control affine model introduced in \eqref{eq:ca-dyn} is amenable to capture also the dynamics of robotic manipulators. For the remainder of this section, we let the state of the robot modeled by \eqref{eq:robotdynmodel} be denoted by $x = [q\tr,\dot q\tr]\tr$ and its input $u=\tau$, so that \eqref{eq:robotdynmodel} can be written in the control affine form \eqref{eq:ca-dyn} with
\begin{equation}
	f(x) = \begin{bmatrix}
		\dot q\\
		-D\inv(q)\big(C(q,\dot q)\dot q + F_v \dot q+g(q)\big)
	\end{bmatrix}
	\label{eq:f}
\end{equation}
and
\begin{equation}
	g(x) = \begin{bmatrix}
		0\\
		D\inv(q)
	\end{bmatrix}.
	\label{eq:g}
\end{equation}

Without loss of generality, in this section we only consider ESB tasks which do not depend explicitly on time. A task $T_i$, defined in terms of the task variable $\sigma_i$ via the CBF $h_i(\sigma_i)$, can be executed by enforcing the constraint $\frac{\partial h_i}{\partial \sigma_i} \frac{\partial \sigma_i}{\partial x} f(x) + \frac{\partial h_i}{\partial \sigma_i} \frac{\partial \sigma_i}{\partial x} g(x) u + \gamma(h_i(\sigma_i))\ge0$.
If $\frac{\partial h_i}{\partial \sigma_i} \frac{\partial \sigma_i}{\partial x} g(x) = 0$,
then the CBF $h_i$ is said to have relative degree higher than 1. At this point, one could make use of exponential CBFs~\cite{nguyen2016exponential} or, more generally, define the following auxiliary CBF (as in \cite{notomista2019persistification}):
\begin{equation}
	\label{eq:hiprime}
	h_i^\prime(\sigma_i,x) = \dot h_i(\sigma_i,x) + \gamma_i (h_i(\sigma_i)),
\end{equation}
and then enforce the following condition:
$\dot h^\prime_i(\sigma_i,x,u) + \gamma(h^\prime_i(\sigma_i)) \ge 0$,
where the term $\dot h^\prime_i(\sigma_i,x,u)$ explicitly depends on $u$ as the dynamical system \eqref{eq:robotdynmodel} is a second-order system.

In general, we define $h_i^\prime$ to be:
\begin{equation}
	\label{eq:hprimegeneral}
	h^\prime_i(\sigma_i,x) = \begin{cases}
		h_i(\sigma_i) &\text{ if $h_i$ has relative degree 1}\\
		\text{Eq.}~\eqref{eq:hiprime} &\text{ if $h_i$ has relative degree 2},
	\end{cases}
\end{equation}
and we compute the control action (i.e., joint torques) by solving the optimization problem in~\eqref{eq:mainQPpriorities} with $h'_i$ and $\gamma'_i$ in place of $h_i$ and $\gamma_i$, respectively.

The computed solution may not be physically implementable in the presence of torque bounds, which are expressed as box constraints on the input optimization variable $u$. However, if one attempts to solve the QP with the additional constraint $\|u\|_\infty \le u_\mathrm{max} < \infty$,
tasks might not be executed as desired, as the optimal control input solution to the optimization program results in large relaxation variables $\delta$. The problem becomes even more serious when multiple safety-critical constraints are present. In fact, by definition, safety-critical constraints should not be relaxed, and, if more than one such constraint is present, the optimization program might be infeasible.

Actuation constraints in multi-task execution have been considered already in the context of Jacobian-based tasks. The so-called \emph{task scaling} procedure was designed to slow down the execution of the task by scaling down joint velocities and accelerations not to exceed the maximum allowed values. However, as a result, the task execution performance may be degraded. Several approaches to this problem have been proposed (see, e.g., \cite{faroni2020inverse} or \cite{lee2022quadratic} for recent solutions). In \cite{Flacco2015}, the authors present an algorithm to select the least possible task-scaling in order to control redundant manipulators in presence of hard joint constraints, in terms of joint positions, velocities, and accelerations.

A constraint-based control strategy for multi-task execution has been presented in \cite{basso2020task}, where the authors consider the robot dynamic model \eqref{eq:robotdynmodel} and include torque bounds in the optimization program required to evaluate the input torques to execute the robotic tasks. Nevertheless, a formal analysis on the quality of task execution in presence of input bounds is lacking. Relaxation variables allow the optimization program to remain feasible at the cost of degrading the performance of task execution, similarly to the task scaling procedure mentioned above.

In this paper, we built upon the CBF-based method proposed in \cite{ames2020integral} to account for input bounds. The solution consists in dynamically extending a system in order to turn the input vector into a state and be able to enforce input constraint using \textit{integral CBFs}, by treating them as state constraints. However, the main result in \cite{ames2020integral} hinges on the feasibility of the formulated QP, which is not guaranteed. Thus, a careful parameter choice has to be made in order to ensure the feasibility of the optimization program at each point in time during online operations.

In this section, we present a systematic approach to find the best possible choice of parameters which are able to ensure that the optimization program defined to synthesize the controller value will always remain feasible, even in presence of input bounds and multiple safety-critical tasks. This will be achieved without degrading the performance of task execution via relaxation variables.

More specifically, the problem of the feasibility of the task execution in presence of input bounds can be solved by an appropriate choice of the class $\mc K$ functions $\gamma_i$ used to define $h_i^\prime$ in \eqref{eq:hiprime}. In the following, we consider linear class $\mc K$ functions, namely $\gamma_i(s) = \Gamma_i s$,
with $\Gamma_i>0$ for all $i$. The problem now becomes that of finding values of the parameters $\Gamma_i$ so that \eqref{eq:robotdynmodel} is always feasible, even in presence of multiple safety-critical tasks and/or input bounds. This objective can be cast as the following semi-infinite program:
\begin{equation}
	\label{eq:maxgamma}
	\begin{aligned}
		\maximize_{u,\Gamma_1,\ldots,\Gamma_M} & \sum_{i=1}^M \Gamma_i\\
		\st & \dot h_i^\prime(\sigma_i,x,u,\Gamma_i) \ge 0 \\
		&\|u\|_\infty\le u_\mathrm{max}\\
		&\Gamma_i\ge0\\
		& \forall x\in \mathcal X_\text{constr}, \forall i\in\{1,\ldots,M\},
	\end{aligned}
\end{equation}
where $\mathcal X_\text{constr} = \bigcup_{j\in\{1,\ldots,M\}}\hspace{-0.15cm}\big\{x\in\R^{n_x}\colon h^\prime_j(\sigma_j(x))=0\big\}$. The maximization is justified by the fact that we aim at finding the largest possible approximation of the feasible set which is affine in $\Gamma_i$ and hence its size grows with $\Gamma_i$. The reason to apply the constraints $\dot h_i^\prime(\sigma_i,x,u,\Gamma_i)\ge0$ only on the union of the boundaries of the zero superlevel sets of the functions $h^\prime_i(\sigma_i)$ follows from Nagumo's theorem \cite{nagumo1942lage}. This allows us to significantly reduce the complexity of solving \eqref{eq:maxgamma}. The latter, moreover, despite being a semi-infinite program, can be computed offline only once in order to obtain optimal values of the parameters $\Gamma_i$.
To summarize, employing the $\Gamma_i$ solution of \eqref{eq:maxgamma} in the definition of $h^\prime_i$ in \eqref{eq:hprimegeneral} has two beneficial effects:
\begin{enumerate}
	\item Tasks need not be relaxed because of torque bounds
	\item The task prioritization and execution optimization program is always feasible even in presence of multiple safety-critical constraints and torque bounds.
\end{enumerate}

Before showing our proposed approach applied to torque-controlled robots, we conclude by emphasizing how the optimization program~\eqref{eq:maxgamma} is solved once and offline to select the parameters to be used subsequently and in an online fashion in the optimization program to synthesize the control inputs with prescribed input bounds.

\begin{example}[Task execution with torque bounds]
	\label{exmp:torquebounds}
	
	Let us compare the behavior of a dynamic robotic manipulator executing an ESB task to control the end-effector position with input torque constraints, using saturation and integral CBFs, respectively. The ESB task consists in driving the end-effector of a 3-link serial manipulator to a desired configuration in the Cartesian space. As a consequence of the fact that there are no joint speed constraints, the solution of \eqref{eq:maxgamma} in terms of $\Gamma_i$ is unbounded above, i.e., an arbitrarily high value of $\Gamma_i$ can be chosen in the definition of $h_i^\prime$,
	\begin{figure*}
		\begin{minipage}{0.5\textwidth}
			\centering
			\begin{minipage}{0.3\linewidth}
				\subfloat[\label{fig:torqueboundsexample:sat:robot}]{\includegraphics[width=\linewidth]{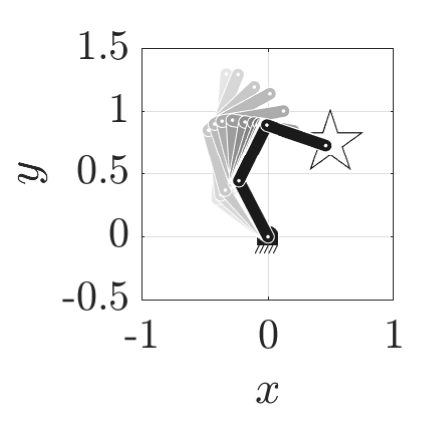}}\hfill
			\end{minipage}
			\begin{minipage}{0.68\textwidth}
				\subfloat[\label{fig:torqueboundsexample:sat:h}]{\includegraphics[trim = {0 26px 0 0}, clip,width=\linewidth]{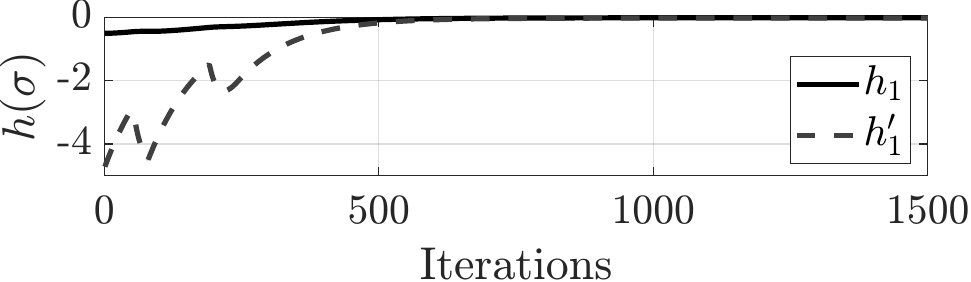}}\hfill
				\subfloat[\label{fig:torqueboundsexample:sat:u}]{\includegraphics[width=\linewidth]{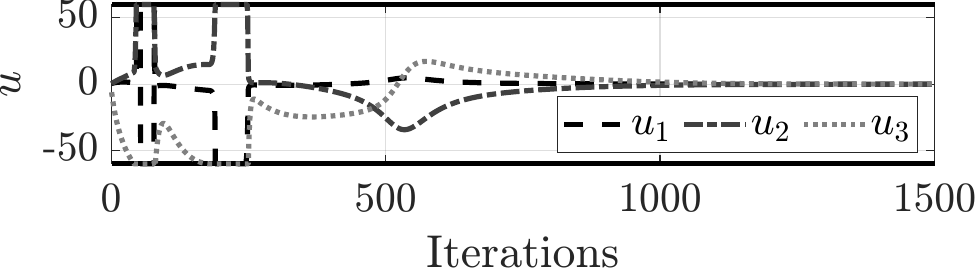}}
			\end{minipage}
		\end{minipage}%
		\begin{minipage}{0.5\textwidth}
			\centering
			\begin{minipage}{0.3\linewidth}
				\subfloat[\label{fig:torqueboundsexample:icbf:robot}]{\includegraphics[width=\linewidth]{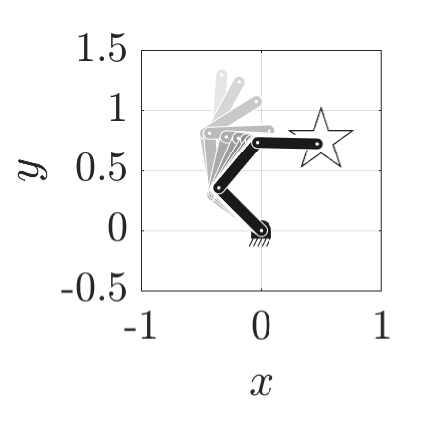}}\hfill
			\end{minipage}
			\begin{minipage}{0.68\textwidth}
				\subfloat[\label{fig:torqueboundsexample:icbf:h}]{\includegraphics[trim = {0 26px 0 0}, clip, width=\linewidth]{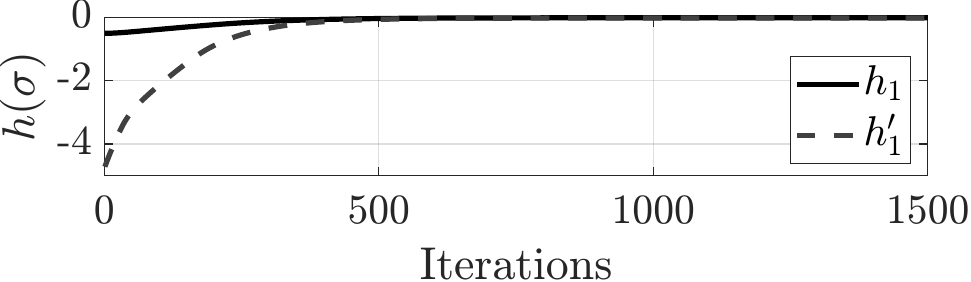}}\hfill
				\subfloat[\label{fig:torqueboundsexample:icbf:u}]{\includegraphics[width=\linewidth]{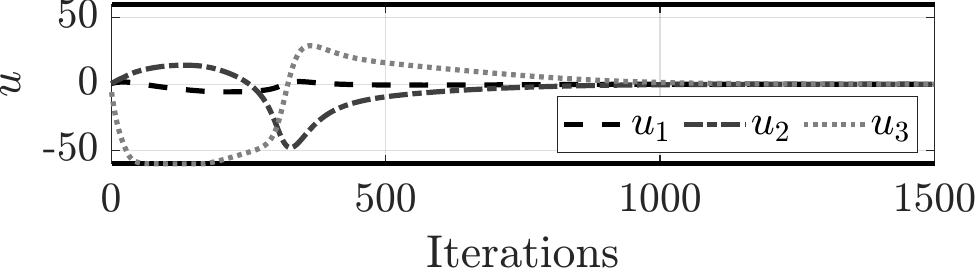}}
			\end{minipage}
		\end{minipage}
		\caption{Enforcing torque bounds using a posteriori saturation (left) \textsc{vs} CBFs on torque input (right). When saturating torque values, the maximum/minimum torques were set to $\pm60$~Nm to make the simulation of the saturated controller stable. On the contrary, using integral CBFs never compromises the stability of the task execution. Furthermore, torques computed using~\eqref{eq:mainQPpriorities} with the CBF in~\eqref{eq:hiprime}, and additional integral CBFs to enforce torque bounds are much smoother compared to those achieved via saturation.
		}
		\label{fig:torqueboundsexample}
	\end{figure*}    

	Figure~\ref{fig:torqueboundsexample} reports the results of two simulations performed saturating input torques after solving the optimization program~\eqref{eq:mainQPpriorities} with the CBF in~\eqref{eq:hiprime} (Figures~\ref{fig:torqueboundsexample:sat:robot}, \ref{fig:torqueboundsexample:sat:h}, and \ref{fig:torqueboundsexample:sat:u}), and enforcing input constraints and task execution constraints holistically using CBFs and integral CBFs, as described in this section (Figures~\ref{fig:torqueboundsexample:icbf:robot}, \ref{fig:torqueboundsexample:icbf:h}, and \ref{fig:torqueboundsexample:icbf:u}). First of all, it is important to mention that the value of the maximum torque (60 Nm) was chosen so to make the simulation of the saturated controller stable. Using integral CBFs, the stability of the task execution itself is never compromised by a maximum value of the input torque.
	
	A second advantage of using the approach described in this section over an a posteriori saturation of the input torques is visible when comparing Figs.~\ref{fig:torqueboundsexample:sat:u} and \ref{fig:torqueboundsexample:icbf:u}. The time evolution of the torques is generally smoother when CBFs are used compared to the torque saturation approach. This is a desirable behavior when dealing with torque-controlled robots as an abrupt change in the torque value may introduce higher order effects which, in turn, may produce practical discontinuities that compromise the stability of the robotic system.
\end{example}

\section{Stability of Prioritized Execution of Extended Set-based Tasks}

We now present results on the stability of ESB tasks with and without prioritized execution stacks. As will be seen, these results are obtained via an extension of the proofs for Jacobian-based tasks, provided in \cite{antonelli2009tro}.

Without loss of generality, we consider the case of controlling the task variable $\sigma(q)$ to zero for all the tasks. We begin by proving stability in executing multiple Jacobian-based tasks without priorities.

In the case of ESB tasks, we carry out computations for the case where the CBFs encoding the tasks do not explicitly depend on the time variable. This choice does not undermine the derived results, as will be pointed out in the following.

\begin{proposition}[Stability of multiple ESB tasks]
	\label{prop:cbf}
	Consider executing $M$ ESB tasks with no priorities by solving the optimization problem in~\eqref{eq:mainQP}, with $\gamma_i(s) = \Lambda_i s$, $\Lambda_i>0$, $\forall i$. Let $\sigma(q) = [\sigma_1(q)\tr, \dots, \sigma_M(q)\tr]\tr$ and $h(\sigma(q)) = [h_1(\sigma_1(q)), \dots, h_M(\sigma_M(q))]\tr$ be the stacked vector of task variables and CBFs encoding the ESB tasks, respectively. Then, $\sigma(q)$ converges to the set
	\begin{equation}
		\label{eq:stableSet}
		\mc S = \left\{\sigma(q) \colon \dhdq\tr \gamma(h(\sigma(q))) = 0 \right\},
	\end{equation}
    where the $i$-th component of $\gamma(h(\sigma(q)))$ is defined to be $\gamma_i(h_i(\sigma_i(q)))$.
\end{proposition}
\begin{proof}
	Applying KKT conditions to the problem~\eqref{eq:mainQP} with the kinematic robot model \eqref{eq:robotkinmodel}, the following equations can be derived:
	\begin{equation}
        \label{eq:kkt_ntasks}
		\dot{q} = -l\dhdq\tr \delta, \qquad \delta = {\Bigg(\underbrace{I+l\dhdq\dhdq\tr}_{A_0}\Bigg)}^{-1} \gamma(h).
	\end{equation}
	To show the stability of the execution of multiple ESB tasks, we consider the Lyapunov function candidate $V(\sigma) = \frac{1}{2}\gamma\left(h(\sigma)\right)\tr \Lambda\inv\gamma\left(h(\sigma)\right)$,
    where $\Lambda = \operatorname{diag}(\Lambda_1, \Lambda_2, \ldots, \Lambda_M)$, $\Lambda\in\mathbb{R}^{M \times M}$. Taking its derivative and using~\eqref{eq:kkt_ntasks} we get
	\begin{equation}
		\label{eq:vdot}
		\begin{aligned}
			\dot{V}(\sigma) &= -l\gamma(h)\tr \Lambda\inv\frac{\partial \gamma}{\partial h} \dhdq \dhdq\tr A_0\inv\gamma(h)\\
			&= -l \gamma(h)\tr \dhdq \dhdq\tr A_0\inv \gamma(h),
		\end{aligned}
	\end{equation}
    since $\frac{\partial \gamma}{\partial h} = \Lambda$. Equation~\eqref{eq:vdot} can be rearranged as follows:
	\begin{equation}
		\label{eq:vdotcontd}
        \begin{aligned}
			\dot{V}(\sigma) &= -l \gamma(h)\tr \dhdq {\Bigg(\underbrace{I + l \dhdq\tr\dhdq}_{A>0}\Bigg)}^{-1}\dhdq\tr \gamma(h)\\
			&= -l \left\|\dhdq\tr \gamma(h)\right\|_{A\inv}^2 \le 0.
		\end{aligned}
	\end{equation}
	By Proposition~3 in \cite{notomista2019optimal}, $\gamma(h)\to\mc S$ in \eqref{eq:stableSet}.
\end{proof}

Next, we present results which demonstrate the stability of multiple Jacobian-based tasks in the presence of prioritizations among the tasks.

In the following, we start by showing the stability in the special case of multiple prioritized \textit{independent} ESB tasks (Corollary~\ref{cor:cbfindeppriority} of Proposition~\ref{prop:cbf}), and then we state and prove the general result of the stability of multiple prioritized ESB tasks (Proposition~\ref{prop:cbfpriority}) execution.

\begin{corollary}
	\label{cor:cbfindeppriority}
	Assume all hypotheses of Proposition~\ref{prop:cbf} hold. Additionally, let all tasks be pairwise independent or orthogonal according to Definition~\ref{def:independent_CBF_tasks}. Then, all tasks will be accomplished.
\end{corollary}
\begin{proof}
	By Proposition~\ref{prop:cbf}, solving the optimization problem~\eqref{eq:mainQP} leads to
    $\gamma(h(\sigma(q(t)))) \to \nullcol\left(\dhdqsmall\tr\right)$
	as $t\to\infty$. Moreover, by Definition~\ref{def:independent_CBF_tasks}, when all tasks are independent, it follows that
    $\rankcol \left(\dhdqsmall\tr\right) = M$,
	$M$ being the number of tasks, and therefore
    $\nullcol\left(\dhdqsmall\tr\right) = \{0\}$.
	Thus, the tasks converge to the set $h(\sigma(q)) \geq 0$, or, equivalently, they will be accomplished.
	By Remark~\ref{rmk:orthogonal}, this is true also for orthogonal tasks.
\end{proof}

The following proposition shows the effect of the prioritization matrix on the execution and completion of possibly dependent tasks. In the proof, we will use selection matrices in order to extract certain columns from a matrix which are required in the derivation. As an example, the matrix $\Sigma = \begin{bmatrix}
        1 & 0 & 0\\
        0 & 0 & 1
    \end{bmatrix}\tr$,
pre-multiplied by a matrix $A\in\R^{3\times3}$, will select its first and third columns and stack them side by side in the $3\times2$-matrix $A\Sigma$.

\begin{proposition}[Stability of multiple prioritized ESB tasks]
	\label{prop:cbfpriority}
	Consider executing a set of multiple prioritized ESB tasks solving the optimization problem in~\eqref{eq:mainQPpriorities}, with $\gamma_i(s) = \Lambda s$, $\Lambda>0$, $\forall i$. Let $\mc I$ be the index set of active constraints of \eqref{eq:mainQPpriorities}. Let $\Sigma_0\in\{0,1\}^{(2M-1)\times|\mc I|}$ and $\Sigma\in\{0,1\}^{M\times|\mc I|}$ be the selection matrices of active constraints and active task constraints, respectively, so that the vector of task variables whose corresponding constraints in \eqref{eq:mainQPpriorities} are active can be expressed as $\bar\sigma(q) = \Sigma\tr \sigma(q)$, and the matrix whose columns are all the columns of the prioritization matrix $K$ corresponding to active task constraints in \eqref{eq:mainQPpriorities} is given by $\bar K =  K\Sigma$. Then, if the matrix $\bar K \Sigma\tr \dhdq$ does not lose rank,
    the vector of task variables $\sigma(q)$ converges to the set
	\begin{equation}
		\label{eq:stableSetPriorities}
		\mc S = \left\{\sigma(q) ~\colon \bar K \Sigma\tr \gamma(h(\sigma(q))) = 0 \right\},
	\end{equation}
	which corresponds to the condition in which all tasks corresponding to active constraints have been executed fulfilling the desired priorities specified by the prioritization matrix $K$.
\end{proposition}
\begin{proof}
	Applying KKT conditions to the optimization problem~\eqref{eq:mainQPpriorities} with the kinematic robot model \eqref{eq:robotkinmodel}, we obtain
	\begin{equation} \label{eq:kkt_n_prioritized_tasks}
		\dot{q} = \frac{1}{2}\begin{bmatrix}\dfrac{\partial h}{\partial q}\tr & 0\end{bmatrix}{\lambda} \qquad
		\delta = \frac{1}{2l}\begin{bmatrix}I & -K\tr\end{bmatrix}{\lambda}
	\end{equation}
	where $\lambda$ is the vector of Lagrange multipliers. For constraints in the active set, the corresponding Lagrange multipliers is strictly positive, i.e., $\lambda_i>0,~\forall i\in\mc I$. Let $\bar\lambda\in\R^{|\mc I|}$ be the vector of non-zero Lagrange multipliers, and $\bbar\lambda$ the vector of Lagrange multipliers that are equal to zero. To find $\bar\lambda$, we consider the dual problem which yields:
	\begin{equation}\label{eq:lambda_sys}
		\bar\lambda = -\frac{1}{2} \bar A_K\inv \bar \gamma_0(h),
	\end{equation}
	where
	\begin{equation}
		\bar A_K = \frac{1}{4l}
		\begin{bmatrix}
			\bar A_0 & \bar K\tr\\
			\bar K & \bar K \bar K\tr
		\end{bmatrix},\quad
		\bar\gamma_0(h) = \Sigma_0\tr \underbrace{\begin{bmatrix}\gamma(h)\\ 0\end{bmatrix}}_{\gamma_0(h)} = \Sigma\tr\gamma(h).
	\end{equation}
	$\bar A_0$ is defined analogously to $A_0$ in~\eqref{eq:kkt_ntasks} as follows:
	\begin{equation}
		\label{eq:A0bar}
		\bar A_0 = I+l\Sigma\tr\dhdq\dhdq\tr\Sigma = I + l\dhdqbar\dhdqbar\tr,
	\end{equation}
	where the quantity $\dhdqbarsmall = \Sigma\tr\dhdqsmall$
	has been introduced for sake of notational compactness, which will become apparent in the following steps.
	
	Let $z=\bar K\Sigma\tr\gamma(h)$ 
	and consider the following candidate Lyapunov function:
	\begin{equation}
		\label{eq:lyapunovfunct}
		V(z) = \frac{1}{2} \|z\|^2.
	\end{equation}
	By the assumption that the index set of active constraints $\mc I$ changes only finitely many times, $V$ is differentiable almost everywhere. Then, taking its time derivative, we obtain:
	\begin{equation}
		\label{eq:vdotprioritiesfirststep}
		\dot{V} = z\tr \dot z = \gamma(h)\tr \Sigma \bar K\tr \bar K \Sigma\tr \Lambda \dhdq \dot{q}.
	\end{equation}
	Substituting $\dot q$ from \eqref{eq:kkt_n_prioritized_tasks} into \eqref{eq:vdotprioritiesfirststep}, and using the expression of $\bar\lambda$ from \eqref{eq:lambda_sys}, we obtain:
	\begin{equation}
		\label{eq:vdotpriorities}
        \begin{aligned}
			\dot{V} &=\frac{\Lambda}{2}\gamma(h)\tr \Sigma \bar K\tr \bar K \Sigma\tr \dhdq \begin{bmatrix}\dhdqbar\tr & 0\end{bmatrix} \lambda\\
			&= \frac{\Lambda}{2}\gamma(h)\tr \Sigma \bar K\tr \bar K \Sigma\tr \dhdq \begin{bmatrix}\dhdqbar\tr & 0\end{bmatrix} \Sigma_0 \bar\lambda\\
			&= -\frac{\Lambda}{4}\gamma(h)\tr \Sigma \bar K\tr \bar K \Sigma\tr \dhdq \dhdqbar\tr \Sigma 4l\tilde A \Sigma\tr \gamma(h)\\
			&= -\Lambda l\gamma(h)\tr \Sigma \underbrace{\bar K\tr \bar K \dhdqbar \dhdqbar\tr \tilde A}_{\tilde A_K} \Sigma\tr \gamma(h),
		\end{aligned}
	\end{equation}
	where $4l\tilde{A}$ is the upper-left block of $\bar A_K\inv$, and the last equality holds owing to the structure of the matrix $\tilde A$.
	
	The matrix $\tilde{A}$ in~\eqref{eq:vdotpriorities} can be decomposed as follows: $\tilde A = \bar A_0\inv + A^\prime$.
	Here, $\bar A_0$ has been defined in \eqref{eq:A0bar} and $A^\prime = \bar A_0\inv \bar K\tr \Bigg(l \bar K \dhdqbar \bar A\inv \dhdqbar\tr \bar K\tr\Bigg)^{-1} \bar K \bar A_0\inv$, 
	where $\bar A\inv = \left(I + l \dhdqbar\tr\dhdqbar\right)\inv$. 
	Notice that $\tilde A$, and in particular $A^\prime$ exists, by the assumption that $\left(\bar K \Sigma\tr \dhdqsmall\right)$ does not lose rank.
	
	Thus, proceeding similarly to \eqref{eq:vdotcontd}, $\tilde A_K$ in \eqref{eq:vdotpriorities} can be simplified to $\frac{\bar K\tr \bar K}{l}$.
	Hence, $\dot V = -\Lambda\gamma(h)\tr \Sigma \bar K\tr \bar K \Sigma\tr \gamma(h) = -\Lambda z\tr z = -\Lambda V(z) \le 0$,
	and, thus, $V(z)\to0$, $z\to0$, and $\sigma(q)\to\mc S$ defined in~\eqref{eq:stableSetPriorities}.
\end{proof}

\begin{remark}[Loss of rank]
    The assumption on full rank of the matrix $\bar K \Sigma\tr \dhdqsmall$ might not always be satisfied. However, if rank loss occurs only at finitely many isolated time instants along the closed-loop trajectory, the condition holds almost everywhere in time. Although this alone does not guarantee convergence to $\mathcal{S}$, the argument remains valid after the final singular encounter, provided that the solution continues beyond it and all required regularity and decay conditions hold thereafter.
    The following example illustrates that passage through a rank-deficient configuration is possible.
\end{remark}

\begin{example}[Loss of rank]\label{ex:rank}
    This example illustrates a loss of rank of the matrix $\bar K \Sigma\tr \dhdqsmall$. 
	Consider a planar Cartesian robot and two planar tasks $T_1$, $T_2$, with priorities $T_1\prec T_2$, consisting in driving the end-effector to the configurations $\sigma_1 = [-1, \, 0]\tr$, $\sigma_2 = [1, \, 0]\tr$, respectively. The task functions have been defined as in~\cite{notomista2020set} with $\sigma_d(t)$ equal to the constant $\sigma_1$ and $\sigma_2$ for task $T_1$ and $T_2$, respectively. In these settings, the tasks are clearly dependent. 
    A simple calculation shows that, for a value of $\kappa = 1000$, the point such that $\rankcol\left(\bar K \Sigma\tr \dhdqsmall\right) < |\mc I|-1$ is $x^* = [- 1.0004,\,0]\tr$, when all constraints are active. This corresponds to a robot configuration such that the gradients of the tasks whose constraints are active are collinear, and appropriately scaled by the value of the chosen parameter~$l$.
	
	Figure~\ref{fig:rank_exp} shows the results of a simulation in which the robot, initially at $x = [-2, \, 0]\tr$, is controlled by the solution of \eqref{eq:mainQPpriorities}. 
    As shown, passing through $x^*$ does not prevent the robot from executing task $T_1$. In particular, from Fig.~\ref{fig:rank_exp:rank} it is possible to see that the loss of rank, evaluated in MATLAB with tolerance $5e^{-3}$, happens at iteration 44, while from Fig.~\ref{fig:rank_exp:h} and Fig.~\ref{fig:rank_exp:u} we notice continuity in the task functions $h_1$ and $h_2$, and in the input $u$ at the same time instant.
\end{example}
\begin{figure}
	\centering
	\subfloat[\label{fig:rank_exp:rank}]{\includegraphics[width=0.32\linewidth]{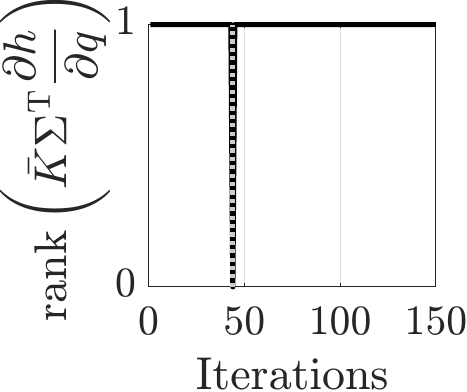}}  
	\hfill\subfloat[\label{fig:rank_exp:h}]{\includegraphics[width=0.32\linewidth]{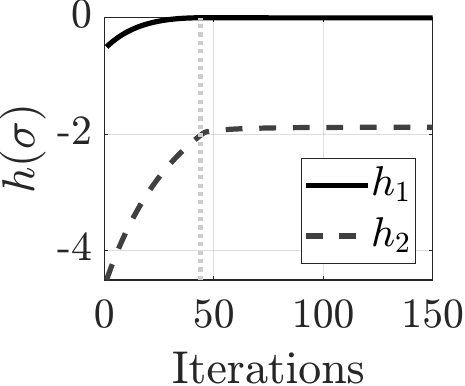}} 
	\hfill\subfloat[\label{fig:rank_exp:u}]{\includegraphics[width=0.32\linewidth]{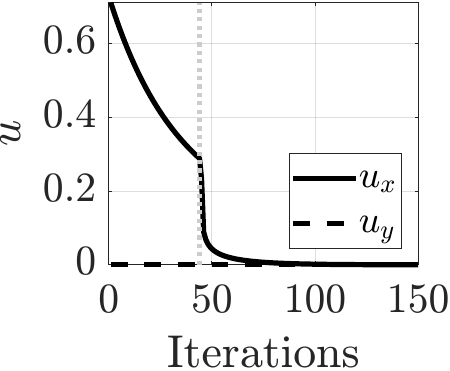}}    
	\caption{Simulation of the scenario in Example~\ref{ex:rank} to show the effects of the loss of rank of the matrix $\rankcol\left(\bar K \Sigma\tr \dhdqsmall\right)$. Figure~\protect\ref{fig:rank_exp:rank} shows the rank evaluation during the simulation. Although rank is lost at iteration 44 (vertical dotted line), Figures~\protect\ref{fig:rank_exp:h} and \protect\ref{fig:rank_exp:u} show that both the task functions $h_1$, $h_2$, and the robot input $u$ remain continuous.}
	\label{fig:rank_exp}
\end{figure}

\section{Automatic and Switching Prioritizations}
\label{sec:switch}

In the previous section, we characterized the stability of the execution of multiple prioritized ESB tasks. This section presents two additional features of the prioritized ESB task execution that can be achieved using the framework proposed in this paper. In Sec.~\ref{subsec:prioritizationmatrixAUTO}, we show how relative priorities between tasks can be obtained in an automatic fashion. This allows us (i) to use fewer parameters to define a prioritized stack and (ii) to execute the closest feasible stack to a required infeasible one. Section~\ref{subsec:switch}, instead, proposes a way to switch between two prioritized stacks of tasks, and provides proofs of input continuity and stability of the controlled system.

\subsection{Automatic Prioritization Stack}
\label{subsec:prioritizationmatrixAUTO}

Assume that $M$ tasks are ordered by priority, so that task $T_i$ has higher priority than task $T_{i+1}$, and consider the following optimization program:
\begin{equation}
	\label{eq:mainQPprioritiesauto}
	\begin{aligned}
		\minimize_{u,\delta,v} &\|u\|^2 + l_\delta\|\delta\|^2  + l_v\|v\|^2\\
		\st & \frac{\partial h_i}{\partial t} + \frac{\partial h_i}{\partial \sigma_i} \frac{\partial \sigma_i}{\partial x} f(x) + \frac{\partial h_i}{\partial \sigma_i} \frac{\partial \sigma_i}{\partial x} g(x) u \\
		&+ \gamma_i(h_i(\sigma_i,t)) \ge - \delta_{i},\quad\forall i \in \{1,\ldots,M\}\\
		& K\delta \le V v.
	\end{aligned}
\end{equation}
The optimization program in \eqref{eq:mainQPprioritiesauto} is similar to the one in \eqref{eq:mainQPpriorities} except for the relaxation of the prioritization constraint. The latter is defined by the following quantities:
\begin{equation}
	K = \begin{bmatrix}
		1 & -\frac{1}{\kappa} & 0 & \ldots & 0\\
		0 & 1 & -\frac{1}{\kappa} & \ldots & 0 \\
		0 & \vdots & \ddots & \ddots & 0 \\
		0 & 0 & \ldots & 1 & -\frac{1}{\kappa}\\
	\end{bmatrix},
\end{equation}
$V = \operatorname{diag}(\kappa^{-1}, \kappa^{0}, \ldots, \kappa^{M-3})$, and the vector of relaxation variables $v=[v_{12},v_{23}, \ldots, v_{M-1M}]\tr\in\R^{M-1}$. The prioritization constraint is relaxed to ensure the feasibility of the optimization program even in the cases of infeasible prioritization stack highlighted in Remark~\ref{rmk:stackfeasibility}. The matrix $V$ is introduced for numerical reasons to make sure that the components of $v$ are of comparable magnitude.

By minimizing $\|v\|^2$, the tasks are executed in a stack which is as close as possible to the prescribed one since $v$ acts as a slack on the stack. As a consequence of the stack constraint relaxation, the priorities encoded by $K$ might not be respected. Finally, note that, the constraint~$ K \delta \le V v$ is an affine relationship between the components of $\delta$ and $v$ therefore~\eqref{eq:mainQPprioritiesauto} is a QP in $u$, $\delta$, and $v$, which can be efficiently solved in an online fashion.

\subsection{Online Switching of Task Prioritizations}
\label{subsec:switch}

The following proposition presents a stack transition method consisting in weighting the control inputs $u_1$ and $u_2$ corresponding to the execution of two stacks prioritized by the prioritization matrices $K_1$ and $K_2$, respectively, in order to transition between them. The same approach is applicable to the case where priorities are enforced using the adaptive prioritization matrix approach treated in Sec.~\ref{subsec:prioritizationmatrixAUTO}.

\begin{proposition}\label{prop:planning_switching}
	Let $u_1$ and $u_2$ be the control inputs corresponding to the execution of the stacks of tasks with prioritization matrices $K_1$ and $K_2$. Let the robot execute the control input
	\begin{equation}
		\label{eq:inputtoswitch}
		u = s(t) u_1 + (1-s(t)) u_2
	\end{equation}
	in the time interval $t\in[0,T]$, where $s(t) = 1-\frac{t}{T}$. Then, the transition from executing the stack of tasks with prioritization matrix $K_1$ to $K_2$ is continuous.
	
	Moreover, at each time instant $t\in[0,T]$, relative priorities that do not change in the initial and final prioritized stack are maintained throughout the switching time interval.
\end{proposition}
\begin{proof}
    Continuity of the controller $u$ in \eqref{eq:inputtoswitch} follows from the continuity of $u_1$, $u_2$ (as shown in \cite{notomista2020set} based on \cite{morris2015continuity}), and $s(t)$.
    
    It remains to show that the relative priority between tasks, whose relative priority remains unchanged, is kept during the transition. To this end, recalling that the component of $\gamma_0(h)$ is different from zero when the corresponding $i$-th task is not fully completed, from Proposition~\ref{prop:cbfpriority}, it follows that when tasks are not completed, one has that $\frac{\partial h}{\partial x} u(t) + \gamma(h) = \delta(t)$, for $u(t)$ and $\delta(t)$ solutions of the optimization program \eqref{eq:mainQPpriorities} at time $t$.
	
	Then, by the definitions of $u_1$ and $u_2$, one can write: $\frac{\partial h}{\partial x} u + \gamma(h) = \frac{\partial h}{\partial x} \big(s(t) u_1 + (1-s(t)) u_2\big) + \gamma(h) = s(t) \left( \frac{\partial h}{\partial x} u_1 + \gamma(h) \right) +\,(1-s(t)) \left( \frac{\partial h}{\partial x} u_2 + \gamma(h) \right) = s(t) \delta_1 + (1-s(t)) \delta_2 =: \Delta$, 
	i.e., $\Delta$ is the convex combination of $\delta_1$ and $\delta_2$. It is assumed that $K_1\delta_1\ge0$ and $K_2\delta_2\ge0$. If task $T_i\prec T_j$ in both stacks, then $\delta_{1,i} < \delta_{1,j}$ and $\delta_{2,i} < \delta_{2,j}$,
	then $\Delta_i(t) = s(t) \delta_{1,i} + (1-s(t)) \delta_{2,i} < s(t) \delta_{1,j} + (1-s(t)) \delta_{2,j} = \Delta_j(t)$, which means that the relative priority between tasks, whose relative priority does not change between initial and final stacks, remains unchanged throughout the transition.
\end{proof}

Besides the amenable minimal invasivity property discussed in the previous remark, another important property of this priority stack switching approach is the resulting stability guarantees on the execution of the prioritized stacks of tasks, as highlighted in the following proposition.

\begin{proposition}
	Consider a robot executing a stack of tasks characterized by the prioritization matrix $K_1$ for $t<0$. In the time interval $t\in[0,T]$, the initial stack is switched to the one characterized by the prioritization matrix $K_2$, using the control input \eqref{eq:inputtoswitch}, where
	\begin{equation}
		\label{eq:softdefinition}
		\begin{cases}
			s(t) \equiv 1 \quad &\forall t\le 0\\
			s(t)\in[0,1] \quad &t\in[0,T]\\
			s(t) \equiv 0 \quad &\forall t\ge T
		\end{cases}
	\end{equation}
	Assuming $\gamma$ is continuously differentiable, and the CBFs encoding the ESB tasks have bounded derivatives, then for $z_2$ defined 
    as $z_2=\bar K_2\Sigma_2\tr\gamma(h)$,
	one has that $z_2(t)\to0$ as $t\to\infty$.
\end{proposition}
\begin{proof}
	Let us start by considering the dynamics of $z_2$ resulting from the application of the input \eqref{eq:inputtoswitch} required to switch prioritized stack:
	\begin{equation}
		\label{eq:z2dynamics}
		\begin{aligned}
			\dot z_2 &= \underbrace{\bar K_2\Sigma_2\tr\frac{\partial \gamma}{\partial h}\frac{\partial h}{\partial q}}_{=:\Psi_2(q)}\dot q =\Psi_2(q)s(t)\dot q_1^\star + \underbrace{\Psi_2(q)(1-s(t))\dot q_2^\star}_{=:(1-s(t))f_\text{cl}(q)}\\
			&=(1-s(t))f_\text{cl}(q) + s(t)\Psi_2(q)\dot q_1^\star,
		\end{aligned}
	\end{equation}
	where $\dot q_1^\star$ and $\dot q_2^\star$ are the solutions of~\eqref{eq:mainQPpriorities} with $K$ equal to $K_1$ and $K_2$, respectively.
	
	Consider the Lyapunov function $V(z_2)=\frac{1}{2}\|z_2\|^2$ defined as in \eqref{eq:lyapunovfunct}. Proposition~\ref{prop:cbfpriority} shows that $V$ is a Lyapunov function for the zero-input system $\dot z_2 = (1-s(t)) f_\text{cl}(q)$,
	owing to the properties of $s(t)$ in \eqref{eq:softdefinition}. Using the results of Proposition~\ref{prop:cbfpriority}, the time derivative of $V$ along the trajectories of the system \eqref{eq:z2dynamics} satisfies the following inequality: $\dot V \le -(1-s(t))\frac{1}{2}\|z_2\|^2 + s(t) z_2\tr \Psi_2(q)\dot q_1^\star$.
	Notice that, by Definition~\ref{def:esbt}, $\partial h / \partial q$ is bounded on a compact set of robot configurations, $q$. Moreover, $s(t)\in[0,1]$ when $t\in[0,T]$ and $\dot{q}_1^*$ is the (bounded) solution of~\eqref{eq:mainQPpriorities}. Therefore, the term $s(t)\Psi_2(q) \dot q_1^\star$ is uniformly bounded. Thus, $V$ is an ISS-Lyapunov function for \eqref{eq:z2dynamics} \cite{sontag2008input}. By Corollary 2.2 in \cite{edwards2000input}, the system \eqref{eq:z2dynamics} is input-to-state stable, i.e., there exist a class $\mc K\mc L$ function $\beta$ and a class $\mc K$ function $\chi$ such that $\|z_2(t)\| \le \beta\left(\|z_2\|,t\right) + \chi\left(\max_{0\le\tau\le t}\|\dot q_1^\star(\tau)\|\right)$. Hence, $\|z_2(T)\|<\infty$, and, by Proposition~\ref{prop:cbfpriority}, $z_2(t)\to0$ as $t\to\infty$.
\end{proof}

This proposition shows that, under mild assumptions (bounded derivative) on the CBFs encoding the ESB tasks to execute, the switch between two prioritized stacks of tasks can be executed without compromising the stability of the robotic system. This will be highlighted in the simulations and experiments reported in Sec.~\ref{sec:simexp}.

So far, we considered only velocity-controlled robots. Nevertheless, in many practical applications, having the ability of controlling the joint torques is required. The next section is devoted to the specialization of the proposed prioritized task execution framework to torque-controlled robotic systems, accounting, in particular, for the presence of torque bounds.

\section{Simulations}
\label{sec:simexp}

In this section, we present the results of simulations performed to showcase the behavior of a robot manipulator executing multiple prioritized tasks controlled using the proposed framework. The execution of three independent tasks is shown in Sec.~\ref{sec:sim1:indip_tasks}, while three dependent prioritized tasks are considered in Sec.~\ref{sec:sim2:dep_tasks}. Sections~\ref{sec:sim3:switching} and \ref{sec:sim4:complex} illustrate the switching behavior between tasks of the same nature and of different nature, respectively. Section~\ref{sec:sim3:switchingdyn} deals with prioritized tasks executed by a torque-controlled robot, and the results of the implementation of the presented framework on the KUKA LBR iiwa 7 R800 manipulator robot is presented in Sec.~\ref{subsec:experiments}.

The robot considered in the simulations in Secs.~\ref{sec:sim1:indip_tasks}-\ref{sec:sim3:switchingdyn} is a planar three-link three-revolute-joint open kinematic chain. Despite its simple geometry, this robotic structure exhibits enough kinematic redundancy with respect to the considered tasks and admits non-trivial singular configurations that can pose numerical challenges. Since these features are shared with generic spatial manipulators, this model provides a meaningful testbed for showcasing the features of our proposed control framework. The length of each of the links is 0.5m. The task functions are defined following Example 2 in~\cite{notomista2020set} where, for each task, $\sigma_d$ is a constant in time. In all the following cases, linear class $\mathcal{K}$ function and Euclidean norm distance are adopted. The code used to run all simulations presented in this section is available online at \url{https://github.com/prisma-lab/nsp_cbf}.

\subsection{Execution of Independent Tasks} \label{sec:sim1:indip_tasks}

\begin{figure}
	\centering
    \begin{minipage}{0.34\linewidth}
	\subfloat[\label{fig:exp1:robot}]{\includegraphics[width=\linewidth]{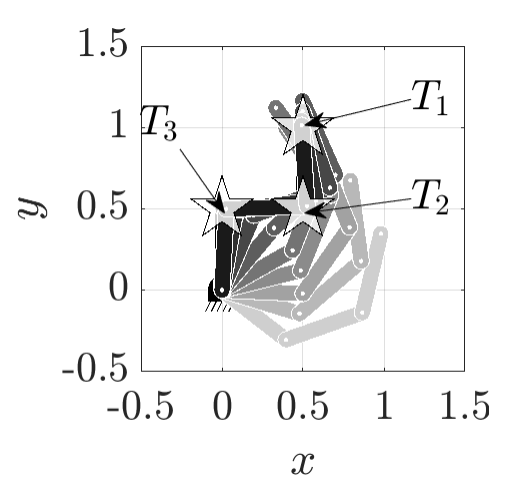}}\hfill
    \end{minipage}
    \begin{minipage}{0.64\linewidth}
	\subfloat[\label{fig:exp1:h}]{\includegraphics[width=0.5\linewidth]{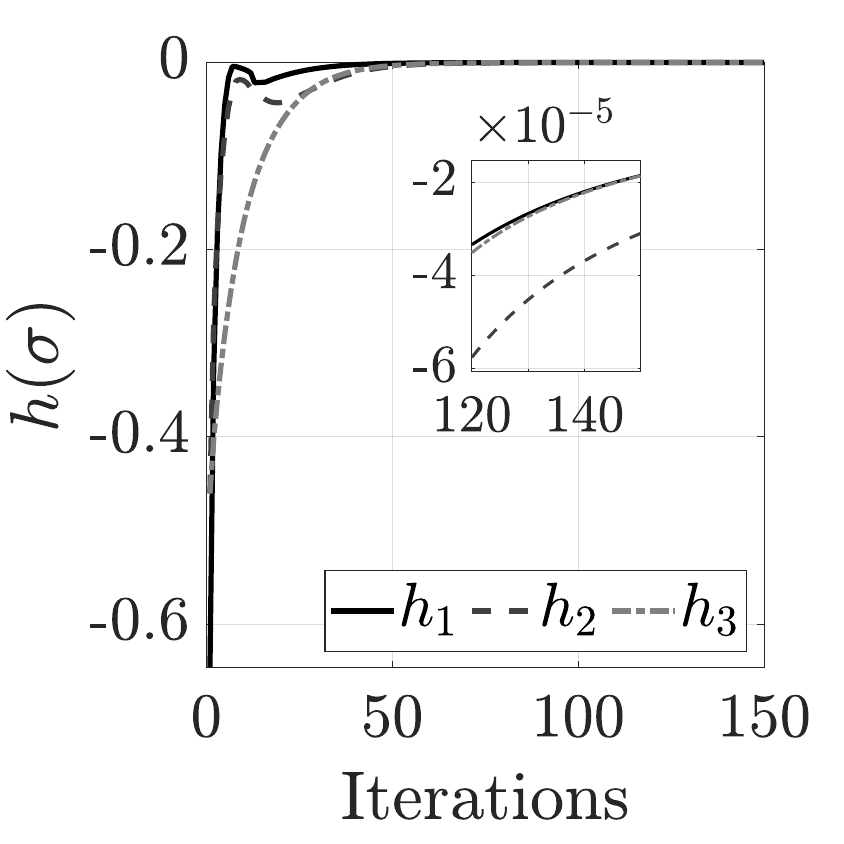}}\hfill
	\subfloat[\label{fig:exp1:V}]{\includegraphics[width=0.5\linewidth]{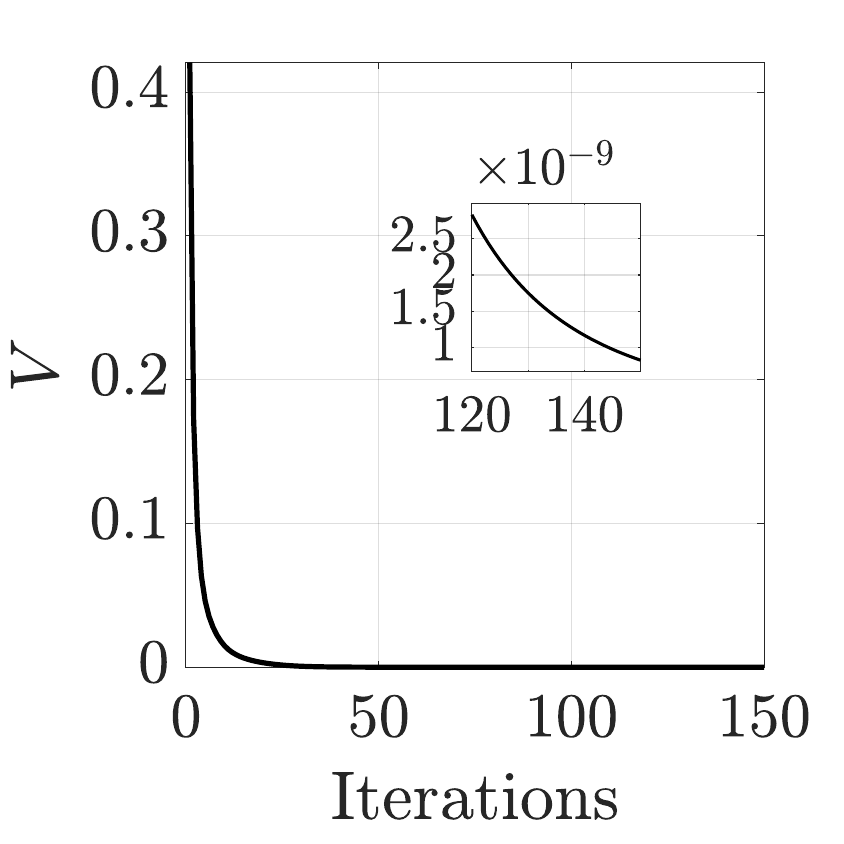}}
    \end{minipage}
	\caption{Execution of 3 independent tasks using the method in Sec.~\ref{subsec:prioritization}. Each task consists in driving the endpoint of link $i$ to the position labeled as $T_i$ and marked with a star. Figure~\protect\ref{fig:exp1:robot} shows the evolution of the configuration of the robot (from light gray to black). Figures~\protect\ref{fig:exp1:h} and \protect\ref{fig:exp1:V} show the evolution of the task functions $h_i$ and the Lyapunov function $V$ in \eqref{eq:lyapunovfunct}, respectively.
	}
	\label{fig:exp1}
\end{figure}

\begin{figure}
	\centering
    \begin{minipage}{0.34\linewidth}
	\subfloat[\label{fig:exp1:robot_exp}]{\includegraphics[width=\linewidth]{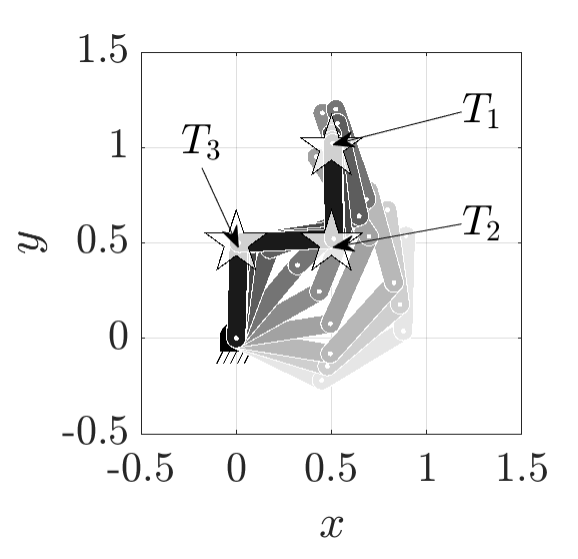}}\hfill
    \end{minipage}
    \begin{minipage}{0.64\linewidth}
	\subfloat[\label{fig:exp1:h_exp}]{\includegraphics[width=0.5\linewidth]{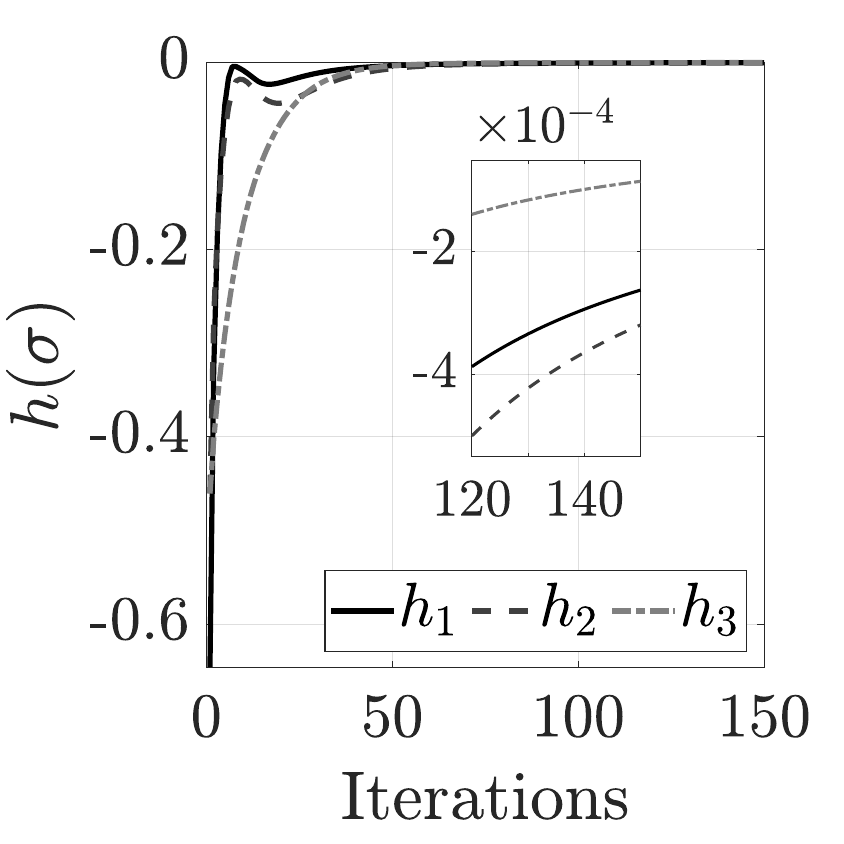}}\hfill
	\subfloat[\label{fig:exp1:v_exp}]{\includegraphics[width=0.5\linewidth]{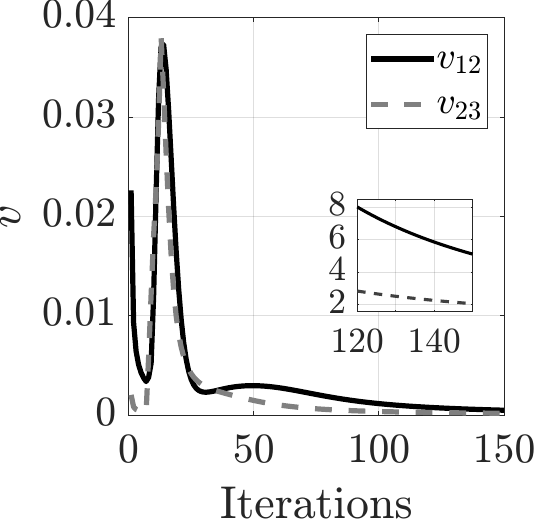}}
    \end{minipage}
    \caption{Execution of 3 independent tasks using the method in Sec.~\ref{subsec:prioritizationmatrixAUTO}. Each task consists in driving the endpoint of link $i$ to the position labeled as $T_i$ and marked with a star. Figure~\protect\ref{fig:exp1:robot_exp} shows the evolution of the configuration of the robot (from light gray to black). Figures~\protect\ref{fig:exp1:h_exp} and \protect\ref{fig:exp1:v_exp} show the evolution of the task functions and of the components of the slack variable $v$ of the prioritization constraint, respectively.}
	\label{fig:exp1_exp}
\end{figure}

This section considers the execution of independent tasks, consisting in controlling the endpoints of each of the three links to reach a desired position in the plane compatible with the geometry of the robot. The three desired positions, depicted as stars in Fig.~\ref{fig:exp1}, are $\sigma_1 = [0.5,\,1]\tr$, $\sigma_2 = [0.5,\,0.5]\tr$, $\sigma_3 = [0,\,0.5]\tr$, respectively. By Proposition~\ref{prop:gener_task_indep}, these three ESB tasks are independent as their corresponding Jacobian-based tasks are independent. 

As per Corollary~\ref{cor:cbfindeppriority}, all the tasks will be executed regardless of their relative priorities. To compare the behavior of the robot executing tasks with a pre-specified stack (as in Sec.~\ref{subsec:prioritization}) with the automatic task prioritization (as in Sec.~\ref{subsec:prioritizationmatrixAUTO}), two simulations are performed. The results of the first one are reported in Fig.~\ref{fig:exp1}. Here the robot is controlled by the optimal input computed by solving QP~\eqref{eq:mainQPpriorities}, with a fixed prioritization stack. By Corollary~\ref{cor:cbfindeppriority}, all tasks are executed as shown in Figs.~\ref{fig:exp1:robot} and \ref{fig:exp1:h}, and the Lyapunov function in \eqref{eq:lyapunovfunct} converges to zero (Fig.~\ref{fig:exp1:V}). It is worthwhile noting that, due to task independence, all functions $h_i$ for $i=1,\,2,\,3$ converge to zero, demonstrating the accomplishment of all the tasks as expected.

While the behavior in Fig.~\ref{fig:exp1} is obtained by letting the robot execute the optimal control input solution to the QP~\eqref{eq:mainQPpriorities}, in Fig.~\ref{fig:exp1_exp} the robot fulfills the same task by executing the solution to the QP~\eqref{eq:mainQPprioritiesauto}. In this case, following the motivation given in Sec.~\ref{subsec:prioritizationmatrixAUTO}, i.e., assuming we do not know whether the designed tasks are independent, we relax the prioritization constraint. The behavior of the robot is shown qualitatively in Fig.~\ref{fig:exp1:robot_exp} and quantitatively in Fig.~\ref{fig:exp1:h_exp}. These figures demonstrate how, \textit{despite the prioritization constraint relaxation}, the robot executes all three tasks, as desired. In Fig.~\ref{fig:exp1:v_exp}, we report the plot of the slack variables---two components of the vector $v$ in \eqref{eq:mainQPprioritiesauto}---which relax the prioritization constraint. As can be seen, the optimization program initially relaxes the prioritization constraints (higher $\|v\|^2$ at the beginning of the simulation), which are asymptotically tightened ($\|v\|^2\to0$ as the simulation iterations increase), effectively realizing the stack prescribed by $K$ in \eqref{eq:mainQPprioritiesauto}.

\subsection{Execution of Dependent Tasks}\label{sec:sim2:dep_tasks}

In this section, we consider dependent tasks. As pointed out in Sec.~\ref{subsec:prioritizationmatrixAUTO}, the approach with a fixed prioritization matrix might not lead to a convergent behavior if the prioritization matrix defining the stack of tasks is not designed respecting the geometry of the tasks themselves. Therefore, for the case of dependent tasks the robot is controlled using the QP~\eqref{eq:mainQPprioritiesauto}, where the prioritization constraint is relaxed.

\begin{figure}
	\centering
    \begin{minipage}{0.34\linewidth}
	\subfloat[\label{fig:exp2:robot}]{\includegraphics[width=\linewidth]{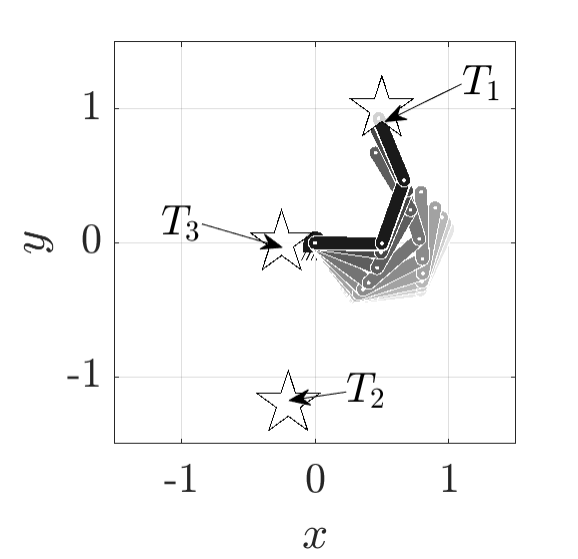}}\hfill
    \end{minipage}
    \begin{minipage}{0.64\linewidth}
	\subfloat[\label{fig:exp2:h}]{\includegraphics[width=0.5\linewidth]{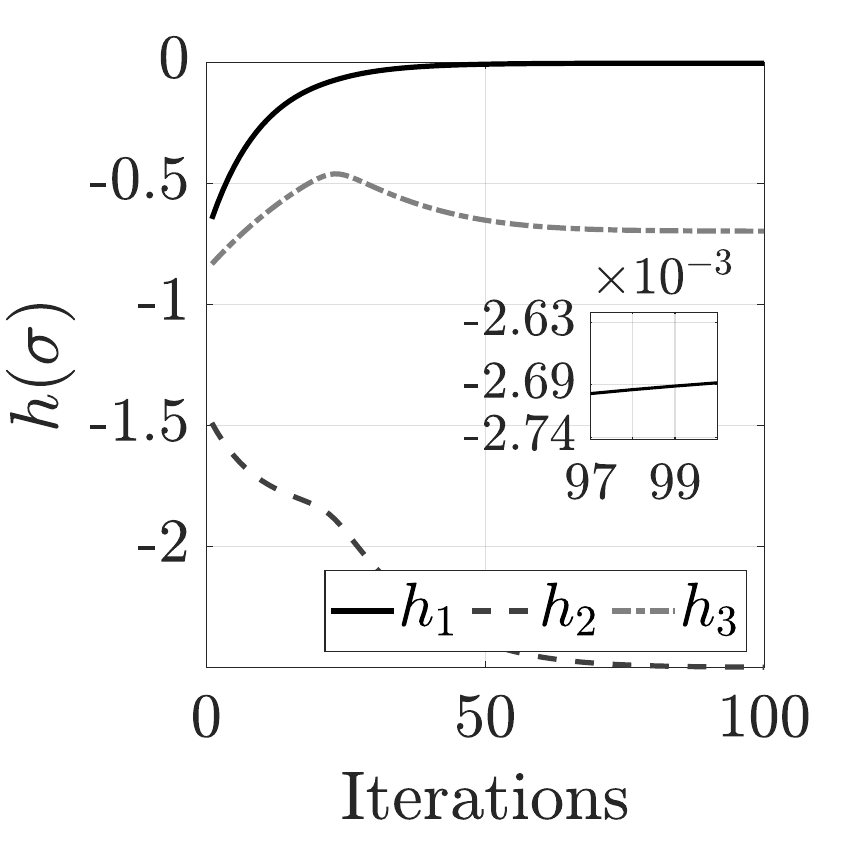}}
    \hfill
	\subfloat[\label{fig:exp2:v}]{\includegraphics[width=0.5\linewidth]{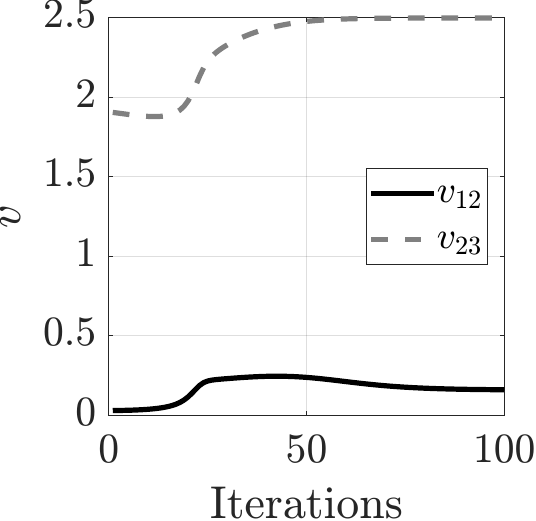}}
    \end{minipage}
	\caption{Execution of 3 dependent tasks consisting in driving the end-effector of the robot to each of the 3 points marked with a star. As in previous simulations, the evolution of the robot configuration is depicted in Fig.~\protect\ref{fig:exp2:robot}. The robot executes task $T_1$ since it has highest priority and all tasks are pairwise dependent. As a result, in Fig.~\protect\ref{fig:exp2:h}, only $h_1$ is driven towards 0, signifying the execution of task $T_1$ only. Figure~\protect\ref{fig:exp2:v} depicts the components of the slack variable $v$ of the prioritization constraint.}
	\label{fig:exp2}
\end{figure}

The three dependent tasks defined in this simulation consist in controlling the end-effector (the endpoint of the third link) of the robot to reach three predefined positions in the plane. The three desired positions, depicted as stars in Fig.~\ref{fig:exp2}, are $\sigma_1 = [0.5,\,1]\tr$, $\sigma_2 = [-0.2,\,-1.2]\tr$, $\sigma_3 = [-0.25,\,0]\tr$, respectively.
The tasks are dependent according to Definition~\ref{def:dependent_CBF_tasks}. In fact, it is not hard to see that there exists a configuration of the robot for which the gradients of the two tasks are aligned.

Proposition~\ref{prop:cbfpriority} shows that the tasks must converge to the set~\eqref{eq:stableSetPriorities}, corresponding to the condition in which all tasks with active constraints are executed respecting their relative priorities. The prescribed prioritization stack is $T_1 \prec T_2 \prec T_3$, with corresponding prioritization matrix equal to
\begin{equation}
    K = \begin{bmatrix}
			1 & -1/\kappa & 0\\
			0 & 1 & -1/\kappa
		\end{bmatrix},
  \label{eq:exp2:prioritizationconstraint}
\end{equation}
where $\kappa=10^3$. With this choice of $\kappa$ and $\sigma_i$, the desired relative priorities are not realizable---because of geometric reasons. \textit{Thanks to the prioritization constraint relaxation}, however, the QP~\eqref{eq:mainQPprioritiesauto} finds the closest prioritization stack that is realizable by the robot.

The results are shown in Fig.~\ref{fig:exp2}. In particular, Fig.~\ref{fig:exp2:robot} depicts the time evolution of the robot configuration, which clearly executes task $T_1$ with highest priority, by moving its end-effector to the point marked with a star and labeled as $T_1$. The corresponding task functions $h_i$ for $i=1,\,2,\,3$ are reported in Fig.~\ref{fig:exp2:h}, where it is possible to see how $h_1(\sigma(t))$ tends to a steady-state value as the simulation iterations increase, while $h_2(\sigma(t))$ and $h_3(\sigma(t))$ converge to bigger values---corresponding to the non-accomplishment of tasks $T_2$ and $T_3$. 

In fact, due to task dependence, not all the values of the functions $h_i$ can converge to zero, i.e., intuitively the three tasks cannot be accomplished simultaneously. Although the resulting errors depend on the task geometry and relaxation, the prioritization constraint allows the error associated with the highest-priority task to be reduced by increasing \(\kappa\).

Figure~\ref{fig:exp2:v} shows the plot of the slack variables of the prioritization stack constraint. As expected, their values do not converge to zero, as the prioritization constraint prescribed by \eqref{eq:exp2:prioritizationconstraint} is not geometrically realizable.

\subsection{Switching Between Dependent Tasks}
\label{sec:sim3:switching}

\begin{figure}
	\centering
	\subfloat[\label{fig:exp3:robot1}]{\includegraphics[width=0.3\linewidth]{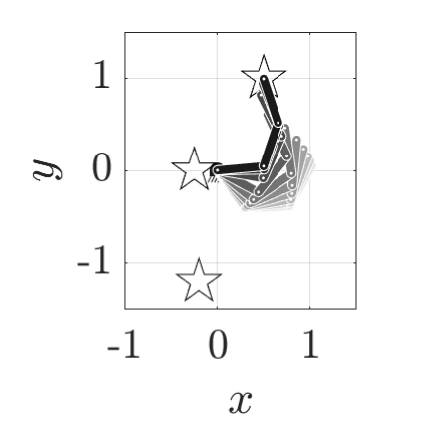}}
	\subfloat[\label{fig:exp3:robot2}]{\includegraphics[width=0.3\linewidth]{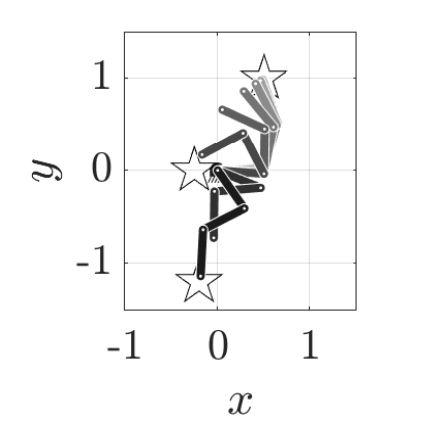}}
	\subfloat[\label{fig:exp3:robot3}]{\includegraphics[width=0.3\linewidth]{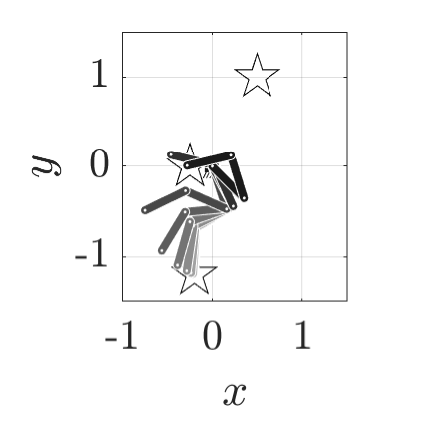}}\\
	\centering 
    \begin{minipage}{0.68\linewidth}
    \subfloat[\label{fig:exp3:h}]{\includegraphics[trim = {0 26px 0 0}, clip, width=0.98\linewidth]{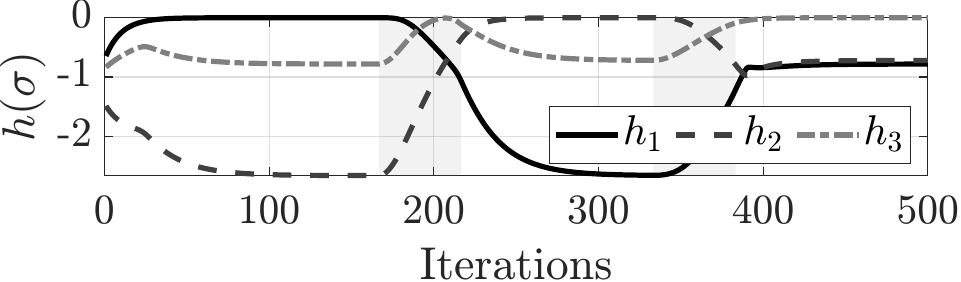}}\\
	\subfloat[\label{fig:exp3:u}]{\includegraphics[width=\linewidth]{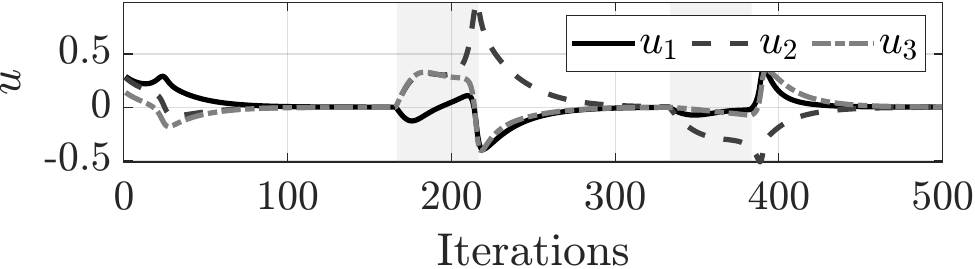}}
    \end{minipage}
	\caption{Priority switching with dependent tasks consisting in driving the end-effector of the robot to each of the 3 points marked with a star. Figures~\protect\ref{fig:exp3:robot1},~\protect\ref{fig:exp3:robot2} and~\protect\ref{fig:exp3:robot3} show the configuration of the robot (evolving from light gray to black in each figure) while it executes task $T_1$, $T_2$, and $T_3$ with highest priority, respectively. Figure~\protect\ref{fig:exp3:h} reports the trajectories of the task functions $h(\sigma)$. Light gray shaded areas denote the intervals during which priorities are swapped.
    Fig.~\protect\ref{fig:exp3:u} shows the continuous joint velocity control input.
    }
	\label{fig:exp3}
\end{figure}

In this and the next sections, we switch our focus to dynamic prioritization stacks where the relative priorities between tasks change over time, and show the behavior of the kinematic and dynamic robot models under the control framework presented in Sec.~\ref{sec:switch}.

In this section, we consider the same tasks of the previous example, which consist in reaching predefined positions in the task space with the robot end-effector. As before, the three desired positions are $\sigma_1 = [0.5,\,1]\tr$, $\sigma_2 = [-0.2,\,-1.2]\tr$, $\sigma_3 = [-0.25,\,0]\tr$, respectively, resulting in dependent tasks. Their relative priorities change according to the following sequence of stacks: 
$T_1 \prec T_2 \prec T_3$ for iterations $[0,\, 166[$, $T_2 \prec T_3 \prec T_1$ for iterations $[166,\, 333[$, and $T_3 \prec T_1 \prec T_2$ for iterations $[333,\, 500]$.

Figure~\ref{fig:exp3} shows data recorded during the simulation. In particular, Figs.~\ref{fig:exp3:robot1}--\ref{fig:exp3:robot3} depict the motion of the robot executing the tasks prioritized based on the sequence of three stacks specified above.
As can be seen, the robot moves sequentially to the point corresponding to the task with highest priority. Figure~\ref{fig:exp3:h} shows the time evolution of the task functions $h_i$ for $i=1,\,2,\,3$. Notice how the value of the function corresponding to the task with highest priority always converges to a value close to zero. Finally, Fig.~\ref{fig:exp3:u} shows the robot input $u=\dot q$, highlighting its continuity during the switching phases.

\subsection{Task Insertion and Removal}
\label{sec:sim4:complex}

\begin{figure}
	\centering
	\begin{minipage}{0.3\linewidth}
		\subfloat[\label{fig:exp4:robot1}]{\includegraphics[trim={0 1.4cm 0 1.2cm}, clip, width=1\linewidth]{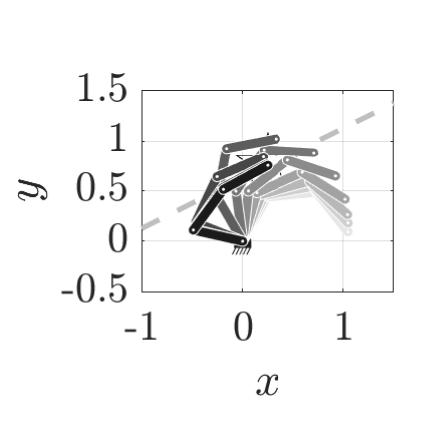}}\\
		\subfloat[\label{fig:exp4:robot2}]{\includegraphics[trim={0 1.4cm 0 1.2cm}, clip, width=1\linewidth]{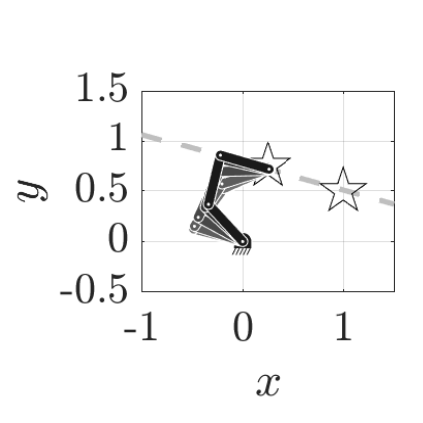}}\\
	\end{minipage}
	\begin{minipage}{0.68\linewidth}
		\subfloat[\label{fig:exp4:h}]{\includegraphics[trim = {0 26px 0 0}, clip, width=\linewidth]{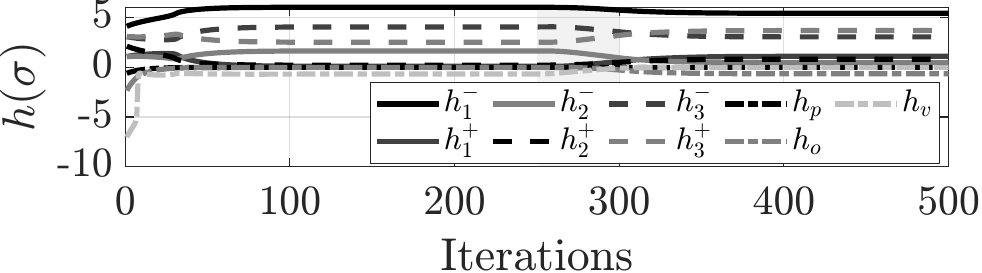}}\\
		\hfill\subfloat[\label{fig:exp4:qdot}]{\includegraphics[width=0.975\linewidth]{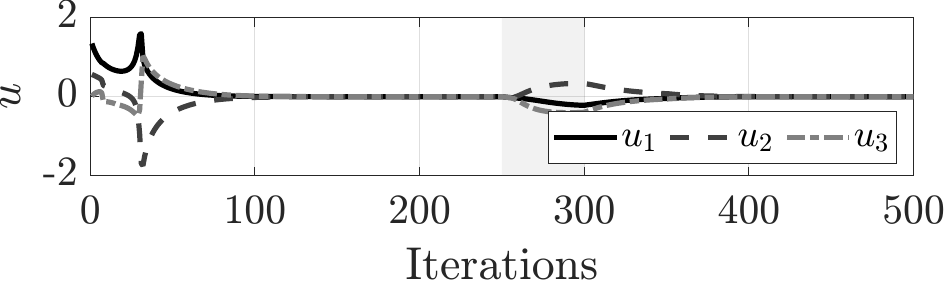}}\\
	\end{minipage}
	\caption{Task insertion/removal and switch between stacks of dependent tasks. The tasks considered in this simulations are joint limit avoidance, end-effector position control, end-effector orientation control, and a look-at-point task in which we control the third link of the robot to the orientation required to aligned it with a desired point of interest.
		Figures~\ref{fig:exp4:robot1} and \ref{fig:exp4:robot2} show the evolution of the robot configurations while executing the two stacks of tasks, respectively. The stars represent the desired end-effector position and the point towards which the third link has to be oriented. The plot of the task functions are reported in Fig.~\ref{fig:exp4:h}, while Fig.~\ref{fig:exp4:qdot} highlights the continuity of the joint velocity control input.
	}
	\label{fig:exp4}
\end{figure}

In this final simulation performed using the kinematic robot model, we demonstrate how the proposed framework can not only handle time-varying prioritization stacks, but also allows for insertions and removals of tasks, dependent and independent alike, in and from a stack being executed. To better illustrate the task insertion/removal capability, in this section we consider tasks of different nature, consisting in remaining within the physical joint limits and orienting the end-effector towards a desired point of interest.
The desired end-effector position is $\sigma = [0.25,\,0.75]\tr$, the desired orientation of the end-effector is $\sigma = \pi/6$, the point to monitor is at $\sigma = [1,\,0.5]\tr$. The joint angle upper and lower bounds are set to $q^+ = [\pi, \, 2/3\pi,\, 2/3\pi]\tr$, $q^- = [-\pi, \, -2/3\pi,\, -2/3\pi]\tr$, respectively. Staying within joint limits is the highest priority task, reaching the desired position is the secondary task, and orientation and look-at-point tasks have lowest priority. 

The simulation is divided into two parts and the switch between task stacks happens in the interval $[250,300]$ iterations. In the first part, the robot has to reach desired position and orientation with its end-effector while staying within joint limits. In the second part, a look-at-point task replaces the orientation task, thus resulting in task removal and insertion. The transition between the task stacks is implemented using the approach in Proposition~\ref{prop:planning_switching}. 

Figure~\ref{fig:exp4} shows data recorded during the simulation. In Figs.~\ref{fig:exp4:robot1} and \ref{fig:exp4:robot2}, the motion of the robot executing the two stacks of tasks, respectively, is depicted. The stars represent the desired end-effector position and the point of interest towards which the third link has to be oriented. The gray dashed lines passing through the desired end-effector position denote the desired orientation in Fig.~\ref{fig:exp4:robot1} and the direction pointing towards the point of interest in Fig.~\ref{fig:exp4:robot2}. Figure~\ref{fig:exp4:h} shows the time evolution of the task functions, where $h_i^+$, $h_i^-$ denote the upper and lower $i$-th joint limit task, respectively. The functions $h_p$, $h_o$, $h_v$ denote the position, orientation and vision task CBF functions. Finally, the continuity of the robot inputs $u=\dot{q}$ is highlighted in Fig.~\ref{fig:exp4:qdot}, during stack switch and task insertion/removal operations.

\subsection{Switching Between Dependent Tasks (Dynamic Model)}
\label{sec:sim3:switchingdyn}

In the previous sections we employed kinematic models of manipulator for the control synthesis. The simulation presented in this section aims at showcasing the proposed task prioritization framework applied to the case of dynamic robot models, controlled using joint torque inputs.

The robot is the same 3-link serial manipulator considered in the previous sections and it is modeled by \eqref{eq:robotdynmodel}. Similar to the previous examples, the two tasks to be executed in a prioritized fashion consist in regulating the end-effector to two different locations in the plane. Thus, they are dependent according to Definition~\ref{def:dependent_CBF_tasks}.

As discussed in Sec.~\ref{sec:dynamics}, when dynamic models are considered, we need to define the auxiliary CBF $h_i^\prime$, as in \eqref{eq:hprimegeneral}. In addition, in this section input constraints are enforced by dynamically extending the model \eqref{eq:robotdynmodel}--\eqref{eq:g} via the input dynamics $\dot u = u^\prime$,
where $x$, $u$, $f$, and $g$ are defined in Sec.~\ref{sec:dynamics}, and $u^\prime$ is the new input to the dynamical system. Proceeding as in \cite{ames2020integral}, the following integral CBF is defined: $h_u(u) = u_\mathrm{max}^2 - \|u\|^2$,
and the constraint $\dot h_u(u,u^\prime) + \gamma_u (h_u(u)) \ge 0$ is enforced to keep the magnitude of the input torques within the desired bound.

\begin{figure}
	\centering
    \begin{minipage}{0.3\linewidth}
	\subfloat[\label{fig:exp3dyn:robot1}]{\includegraphics[width=\linewidth, trim={0 1.4cm 0 1.2cm}, clip]{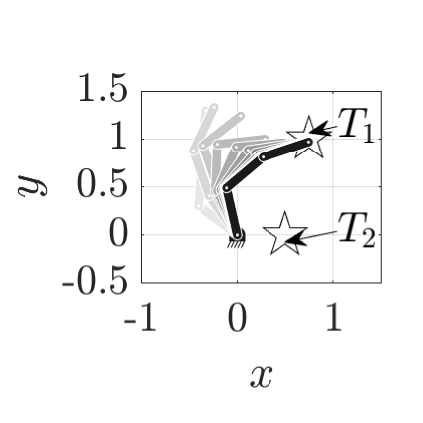}}\hfill
	\subfloat[\label{fig:exp3dyn:robot2}]{\includegraphics[width=\linewidth, trim={0 1.4cm 0 1.2cm}, clip]{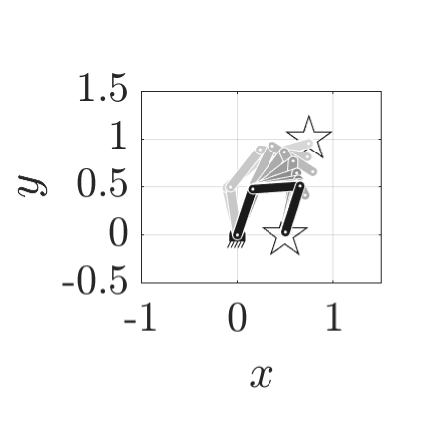}}
    \end{minipage}
    \begin{minipage}{0.68\linewidth}
	\subfloat[\label{fig:exp3dyn:h}]{\includegraphics[trim = {0 26px 0 0}, clip, width=\linewidth]{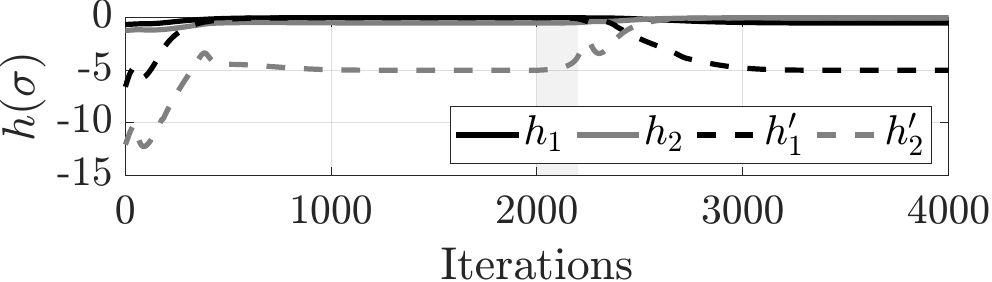}}\\
	\hfill\subfloat[\label{fig:exp3dyn:u}]{\includegraphics[width=0.99\linewidth]{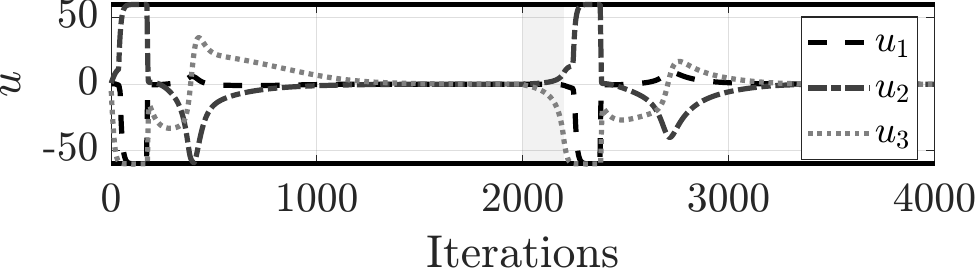}}
    \end{minipage}
	\caption{Priority switching with dependent tasks and torque-controlled robot. Similarly to its kinematic counterpart in Fig.~\ref{fig:exp3}, Fig.~\ref{fig:exp3dyn:robot1} and Fig.~\ref{fig:exp3dyn:robot2} show the evolution of the robot configurations. In Fig.~\ref{fig:exp3dyn:h}, instead, besides the trajectory of $h_1$ and $h_2$, the functions $h_1^\prime$ and $h_2^\prime$ are displayed. 
    Fig.~\ref{fig:exp3dyn:u}, the continuity of the input torques required to switch between the two stacks of tasks is shown, alongside the minimum and maximum values they can attain.}
	\label{fig:exp3dyn}
\end{figure}

Figure~\ref{fig:exp3dyn} shows the result of the execution of the two tasks, whose priority is swapped halfway through the experiment. In particular, Figs.~\ref{fig:exp3dyn:robot1} and \ref{fig:exp3dyn:robot2} show the motion of the robot while executing the prioritized stacks $T_1 \prec T_2$ and $T_2 \prec T_1$, respectively. The positions to which the end-effector is to be regulated are depicted as stars: the top position corresponds to task $T_1$, while controlling the end-effector to the bottom position is equivalent to executing task $T_2$. The values of the CBFs $h_1$, $h_2$, $h_1^\prime$, and $h_2^\prime$ used to execute the two tasks are plotted in Fig.~\ref{fig:exp3dyn:h}. As can be seen, the CBFs corresponding to the task executed with the highest priority are driven to zero. Finally, the joint torques, $u$, obtained by integrating $u^\prime$ 
and used to control the robot, are reported in Fig.~\ref{fig:exp3dyn:u}, where we observe that torques are continuous during the priority switch. Furthermore, the torque values, depicted as thick black lines, never exceed the upper and lower bounds, $u_\mathrm{max}=60$~Nm. Notice that to further smooth out the $u$ signal, integral CBFs may be used (see discussion in Fig.~\ref{fig:torqueboundsexample}).

\begin{figure}
	\centering
    \begin{minipage}{0.3\linewidth}

	\subfloat[\label{fig:exp3dynHQP:robot1}]{\includegraphics[width=\linewidth, trim={0 1.4cm 0 1.2cm}, clip]{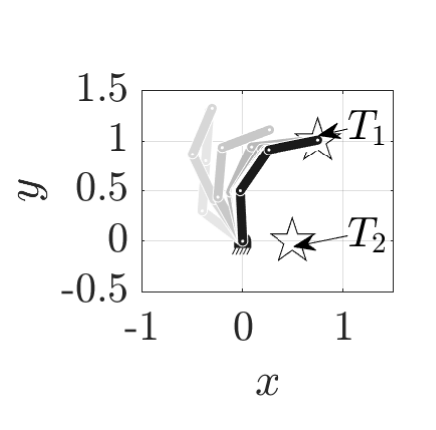}}\hfill
	\subfloat[\label{fig:exp3dynHQP:robot2}]{\includegraphics[width=\linewidth, trim={0 1.4cm 0 1.2cm}, clip]{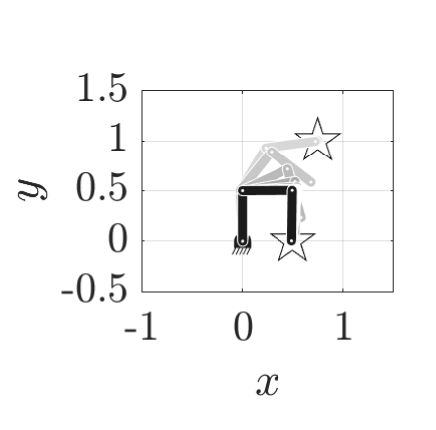}}
    \end{minipage}
    \begin{minipage}{0.68\linewidth}
	\hfill\subfloat[\label{fig:exp3dynHQP:u}]{\includegraphics[width=0.99\linewidth]{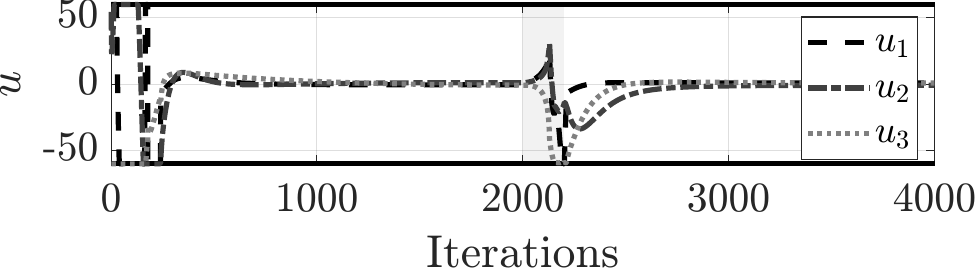}}
    \end{minipage}
	\caption{Priority switching with dependent tasks and torque-controlled robot using the HQP method~\cite{kim2019continuous}. Similarly to its counterpart in Fig.~\ref{fig:exp3dyn}, Fig.~\ref{fig:exp3dynHQP:robot1} and Fig.~\ref{fig:exp3dynHQP:robot2} show the evolution of the robot configurations. In the first half of the experiment, task $T_1$ has highest priority, in the second half of the experiment, the priorities are swapped. Finally, Fig.~\ref{fig:exp3dynHQP:u}, the input torques required to switch between the two stacks of tasks is shown, alongside the minimum and maximum values they can attain.
	}
	\label{fig:exp3dynHQP}
\end{figure}

\subsubsection{Comparison with Hierarchical Quadratic Programming}

In the following, we compare the performance of the Hierarchical Quadratic Programming (HQP) method presented in~\cite{kim2019continuous} along the execution of the two prioritized tasks, whose priority is swapped half way through the experiment. We have chosen this among other existing approaches (including \cite{kanoun2011kinematic,escande2014hierarchical,somani_task_2016,aertbelien_etasletc_2014}), since, just like our proposed approach, it allows for switching stacks, insertion and removal of tasks, and it is amenable to torque control. Figure~\ref{fig:exp3dynHQP} shows data recorded during the simulation to be compared with Fig.~\ref{fig:exp3dyn}. The motion of the robot executing the two stacks of tasks, is depicted in Figs.~\ref{fig:exp3dynHQP:robot1} and \ref{fig:exp3dynHQP:robot2}, respectively. The robot inputs $u$ are highlighted in Fig.~\ref{fig:exp3dynHQP:u}. The HQP implementation has been tuned to achieve a robot behavior as similar as possible to the previous simulation presented in Sec.~\ref{sec:sim3:switchingdyn}. Despite demonstrating similar performance, the HQP method requires the solution of multiple QP problems during transition. In particular, for these simulation settings the solution of $5$ QP problems is required for the HQP method, each with $2N + 2$ optimization variables and $3N + 2$ constraints. Our method instead requires the solution of only $2$ QP problems each with $N + 2$ optimization variables and $2N + 1$ constraints. Given the polynomial computational complexity with respect to the number of optimization variables and constraints of interior-point methods used to solve each convex QP, the approach proposed in this paper requires a lower computational effort. Nevertheless, it is worthwhile pointing out that this computational advantage comes at the cost of relaxing the execution of higher-priority tasks too, rather than having strict priorities as in the HQP.

\section{Robot Experiments}
\label{subsec:experiments}

\begin{figure}
    \centering
    \includegraphics[trim={2cm 1cm 3cm 3cm}, clip, width=0.6\linewidth]{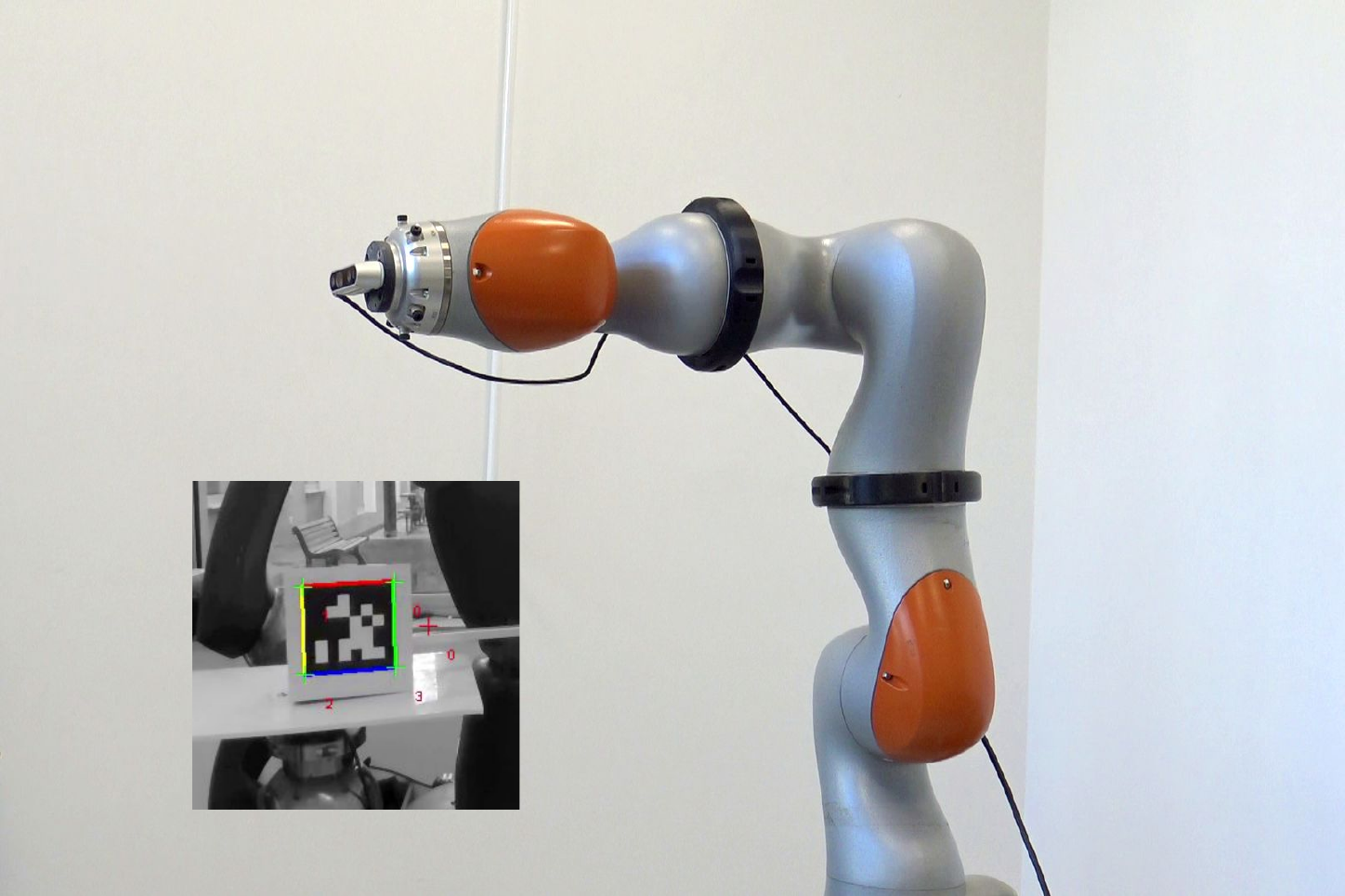}
    \caption{Experimental setup. The KUKA LBR iiwa 7 R800 robot with a RealSense T435 camera at its end-effector. The superimposed picture shows the camera image overlayed with AprilTag tracking marks. Full experiment video available online at \url{https://youtu.be/u_5v52-q2CQ}.}
    \label{fig:setup}
\end{figure}

In this section, we present the results of experiments performed on a real 7-DOF anthropomorphic manipulator. We show how the proposed framework effectively executes and prioritizes a time-varying stack comprising both dependent and independent ESB tasks, as well as safety-critical tasks.

The experimental setup consists of the KUKA LBR iiwa 7 R800 robot with a RealSense T435 camera mounted on its tip (see Fig.~\ref{fig:setup}).
The manipulator is torque-controlled with a \textit{Fast Robot Interface} (FRI) client application~\cite{SchreiberTheFR}, and runs with a command sending period of $T_{cmd}=3$~ms. The torque $\tau_\text{FRI}\in\mathbb{R}^7$, used to track the desired joint positions and velocities, is computed as
\begin{equation}
    \tau_{\text{FRI}}(t) = k_u\,(u(t) - \dot{q}(t)) + k_q\,(q_d(t) - q(t)),
    \label{eq:experiment:tau}
\end{equation}
where $ u\in \mathbb{R}^7$ is obtained solving the QP \eqref{eq:mainQPprioritiesauto}, $q_d$ by integrating the desired joint velocities $u(t)$, and the control gains were chosen to be $k_u=100$ and $k_q=36$. The QP is solved online in a Robot Operating System (ROS) node employing the OSQP library~\cite{osqp}, and evaluating the CBFs using measured joint positions $q$. A middle-layer ROS node acts as an interface between ROS and the FRI client.

To illustrate the proposed framework, we consider tasks of different natures that are dynamically inserted and removed in the prioritized stack. 
The following ESB tasks perform position control: (i) task $T_{p_i}$ consists in reaching a desired position $p_i$ with the end-effector in the Cartesian space, and (ii) task $T_{z_3}$ drives the third link to a desired height along the $z$ axis in the operational space. In addition, a vision task, $T_v$, is used to orient the camera mounted on the end-effector toward a desired point of interest---the AprilTag shown in Fig.~\ref{fig:setup}. To this end, the vision task was achieved using the CBF $h_v(s) = -0.5 k_v\|s-s_d\|^2$, where $k_v$ is a strictly positive scalar and $s, s_d\in\mathbb{R}^3$ are the actual and desired bearing directions, respectively. The vector $s$ is computed as $s = {p_{v,c}^c}/\left({\|p_{v,c}^c\|}\right)$,
where $p_{v,c}^c$ is the relative position of the tag object with respect to the camera, expressed in the camera frame. $s_d$ is its corresponding desired value and denotes the direction in which the object should be placed with respect to the camera. In the experiments reported below, the tag was fixed with respect to the robot's base frame; however, the video of the full experiment available online at \url{https://youtu.be/u_5v52-q2CQ} also shows the case of a moving tag.
We selected $k_v=10$, $s_d = \left[0, 0, 1\right]\tr$.

The end-effector positioning task was encoded in the CBF $h_p(p_e) = -0.5 k_p\|p_e - p_i\|^2$, while the task $T_{z_3}$ with $h_z(p_{z_3}) = -0.5 k_z\|p_{z_3} - z_3\|^2$ where $k_p$ and $k_z$ are strictly positive scalars, $p_e\in\mathbb{R}^3$ is the end-effector Cartesian position, and $p_{z_3}\in\mathbb{R}$ is the height of the third link in the operational space. In the experiments, the desired end-effector position is switched between $p_1 =[0.3, 0.2, 0.8]\tr$m and $p_2=[0.3, -0.2, 0.7]\tr$m; the desired height is set to $z_3=0.45$~m, all expressed in the world frame. We employed $k_p=1$ and $k_z=100$.

The stack comprises also safety-critical set-based tasks to accomplish joint limits avoidance. 
The joint limits task was achieved using the CBF $h^{\pm}_j(q_j) = -k_{q_j}(q_j^+-q_j)(q_j-q_j^-) / (q_j^+ - q_j^-)^2$, $j = 1, \dots, 7$, where $k_{q_j}$ is a strictly positive scalar, and $q_j$ is the $j-$th joint value.
As these tasks are safety critical, they are not relaxed by slack variables.

The experiment is divided into five parts by four stack transitions that happen at times $5\,$s, $35\,$s, $65\,$s, $86\,$s and consists in executing the following sequence of stacks: Empty stack in the time interval $[0,\, 5 [$~s; $T_{p_1} \prec T_v$ in $[5,\, 35 [$~s; $T_{p_1} \prec T_{z_3} \prec T_v$ in $[35,\, 65 [$~s; $T_{p_1} \prec T_{p_2} \prec T_v$ in $[65,\, 86 [$~s; and $T_{p_2} \prec T_{p_1} \prec T_v$ from $86$~s onwards.

Transitions are handled employing the method in Proposition~\ref{prop:planning_switching}, with length of the switching time interval $T=1.5$~s. The switching phases---either due to insertion/removal of tasks or to a change of their relative priorities---are highlighted in Figs.~\ref{fig:experiment:h}, \ref{fig:experiment:uq}, and \ref{fig:experiment:slack} by light gray shaded areas. 

\begin{figure}
    \centering
    \begin{minipage}{1\linewidth}
    \hfill\subfloat[\label{fig:experiment:hjl}]{\includegraphics[width=0.49\linewidth]{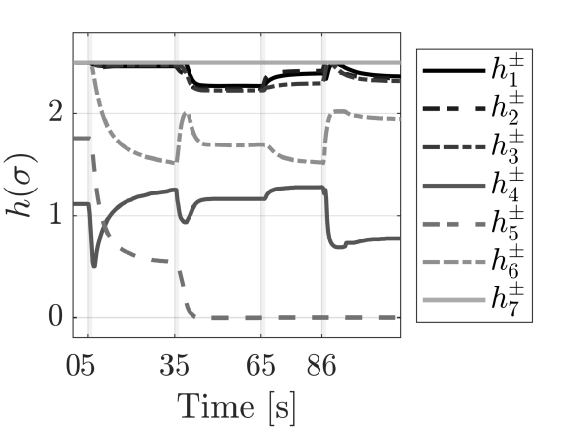}}\hfill
    \subfloat[\label{fig:experiment:ht}]{\includegraphics[width=0.49\linewidth]{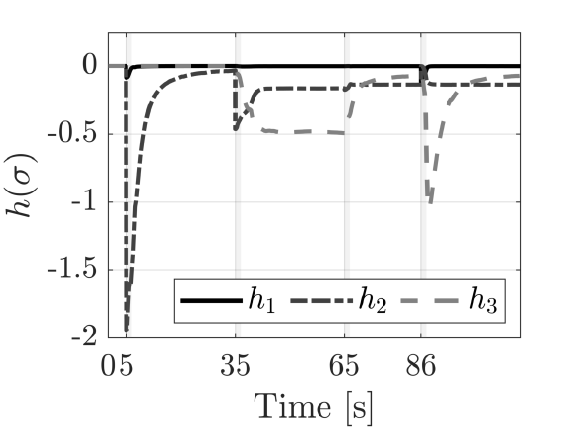}}
    \end{minipage}
    \caption{Values of the CBFs encoding the ESB tasks performed in the experiment. (a) Time evolution of the joint limits avoidance functions $h^{\pm}_{i}$ for $i=1,...,7$ for both the upper and lower $i-$th limits. (b) Evolution of ESB task functions $h_i$ for $i=1,2,3$, ordered according to their priority.
    Light gray shaded areas denote transition intervals.}
    \label{fig:experiment:h}
\end{figure}

\begin{figure}
    \centering
    \begin{minipage}{1\linewidth}

    \subfloat[\label{fig:experiment:u}]{\includegraphics[width=0.49\linewidth]{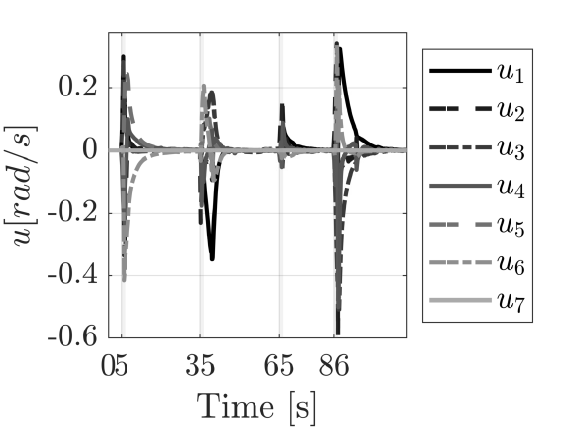}}\hfill
    \subfloat[\label{fig:experiment:q}]{\includegraphics[width=0.49\linewidth]{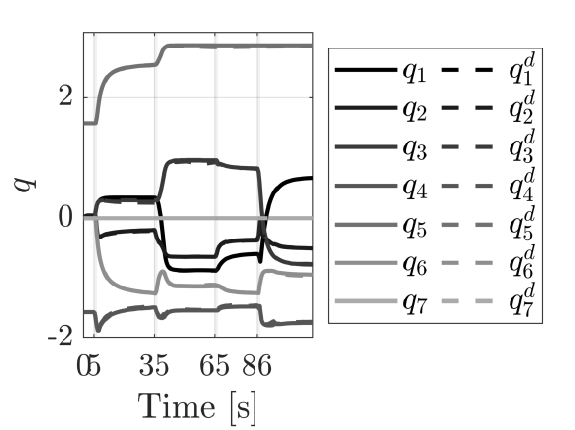}} 
    \end{minipage}
    \caption{Joint velocity control input $u$, desired and measured joint positions recorded over the experiment. (a) Time evolution of the joint velocity control input, highlighting switching time instants and transitions phases (shaded in light gray). (b) Desired and measured joint positions.}
    \label{fig:experiment:uq}
\end{figure}
\begin{figure}
    \centering
    \begin{minipage}{1\linewidth}
    \hfill\subfloat[\label{fig:experiment:d}]{\includegraphics[width=0.49\linewidth]{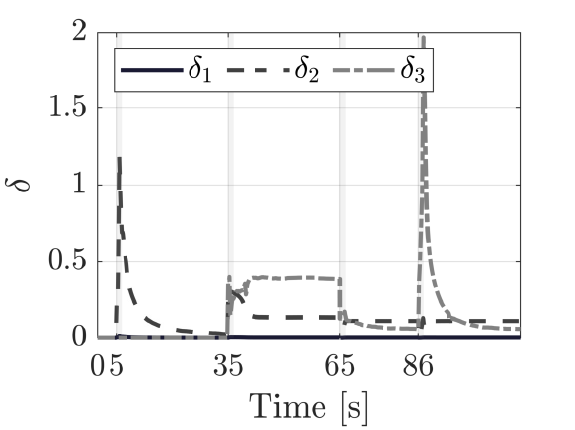}}\hfill
    \subfloat[\label{fig:experiment:v}]{\includegraphics[width=0.49\linewidth]{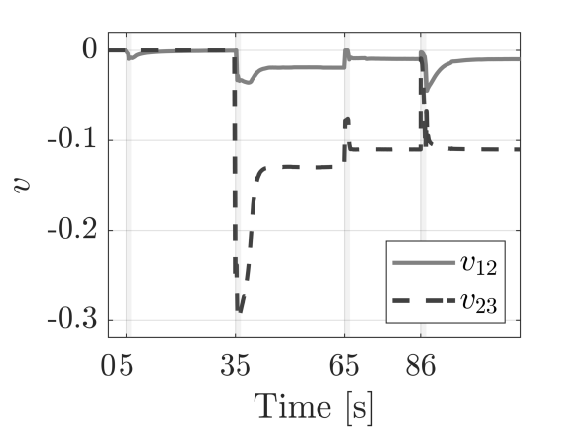}}
    \end{minipage}
    \caption{Time evolution of the QP slack variables during the experiment.(a) Time evolution of $\delta_i$ for $i=1,2,3$ that respectively relax the constraint of the task with $i-$th priority. (b) Time evolution of the components of the slack variable $v_{ij}$ for $i,j=1,2,3$ of the prioritization constraint. Light gray shaded areas denote transition intervals.}
    \label{fig:experiment:slack}
\end{figure}

Figure~\ref{fig:experiment:hjl} shows the CBFs $h^{\pm}_{i}$ for $i=1,\dots,7$ of the joint limits avoidance task. All functions remain non-negative, preventing joint-limit violation even if a limit value is reached, as happens to the fifth joint around $t=40$s ($h^{\pm}_5$ becomes zero). Figure~\ref{fig:experiment:ht} shows the task functions $h(\sigma)$, encoding the ESB tasks, that should be driven to zero. However, as can be seen in Fig.~\ref{fig:experiment:ht}, this does not happen whenever a physical limit is reached or when the prioritized stack is infeasible. In particular, in the interval $[5,35[\,$s, the stack is composed of two independent tasks, $T_{p_1}$ and $T_v$, and the functions $h_1$ and $h_2$ are driven to zero, achieving faster convergence on the highest priority task (see Fig.~\ref{fig:experiment:h}). At time $t=35\,$s, the independent task $T_{z_3}$ is added, but the task functions do not reach zero due to the hit of the fifth joint limit. The precedence relation $T_{p_1} \prec T_{z_3} \prec T_v$ is respected at steady state ($|h_1(\sigma)|<|h_2(\sigma)|<|h_3(\sigma)|$) even if the prioritization constraints are relaxed.

At time $t=65$s, task $T_{z_3}$ is replaced by $T_{p_2}$. Tasks $T_{p_1}$ and $T_{p_2}$ are dependent, and only the first will be executed.
Moreover, the relaxation of the constraint $T_{p_2}\prec T_v$ enables the decrease of the vision task function. Finally, at time $t=86$s, $T_{p_1}$ and $T_{p_2}$ are swapped, with analogous considerations. Computed joint velocities $u$ are plotted in Fig.~\ref{fig:experiment:u} highlighting the continuity of the control input during stack switch and insertion/removal operations. Figure~\ref{fig:experiment:q} shows the desired and measured joint positions, illustrating the good accuracy achieved with the robot kinematic model, the control inputs in Fig.~\ref{fig:experiment:u}, and the control law~\eqref{eq:experiment:tau}. Finally, Figs.~\ref{fig:experiment:d} and \ref{fig:experiment:v} depict the time evolution of the slack variables of the task execution constraints and the prioritization stack constraint, respectively. Together with the time evolution of the CBF values encoding the executed tasks reported in Fig.~\ref{fig:experiment:h}, these figures show how the constraint-based task execution and prioritization framework presented in this paper effectively achieves the desired robot behavior.

\section{Conclusions}
\label{sec:conclusions}

In this paper we analyzed the use of \emph{extended set-based (ESB) tasks} to encode the execution of robotic tasks via the forward invariance and asymptotic stability of appropriately defined sets of the robot state space. The concepts of orthogonality, independence, and dependence defined for Jacobian-based robotic tasks are shown to naturally extend to ESB tasks, carrying a similar meaning in terms of the possibility of concurrently executing multiple tasks. The control framework developed in the paper allows for an effective way of executing multiple ESB tasks and for a flexible way of prioritizing them in time-varying stacks. Theoretical results demonstrated stability properties of the prioritized task execution framework, giving guarantees on the expected robot behavior. Moreover, extensive simulations showcased several features and use cases suitable for the proposed approach. Furthermore, experiments performed on a real manipulator demonstrated how the developed control strategy is suitable to be employed in the tight loop under real time constraints. 

\bibliographystyle{ieeetr}
\bibliography{bib/IEEEabrv,bib/nsp_cbf}

\end{document}